\documentclass[envcountsame]{llncs}
\usepackage[T1]{fontenc}
\usepackage{amsmath}
\usepackage{amssymb}
\usepackage[noadjust]{cite}
\usepackage[retainorgcmds]{IEEEtrantools}
\usepackage{enumerate}
\usepackage{algorithm}
\usepackage[noend]{algorithmic}
\usepackage{centernot}
\usepackage{color}
\usepackage{soul}
\usepackage[absolute]{textpos}

\usepackage{tikz}
\usetikzlibrary{arrows}
\usetikzlibrary{automata}
\usetikzlibrary{fit}
\usetikzlibrary{positioning}
\usetikzlibrary{calc}

\usepackage[textwidth=3.5cm, shadow, disable]{todonotes}

\floatname{algorithm}{Procedure}

\let\originalleft\left
\let\originalright\right
\renewcommand{\left}{\mathopen{}\mathclose\bgroup\originalleft}
\renewcommand{\right}{\aftergroup\egroup\originalright}

\spnewtheorem{fclaim}{Faulty Claim}{\itshape}{\rmfamily}

\newcommand*{\sqcupsqcap}{{\sqcup}{\mkern-3mu \sqcap}}
\newcommand*{\sqcapsqcup}{{\sqcap}{\mkern-3mu \sqcup}}

\let\doendproof\endproof
\renewcommand\endproof{~\hfill\qed\doendproof}

\newcommand{\qedhere}{%
  \begingroup \let\mathqed\math@qedhere
    \let\qed@elt\setQED@elt \QED@stack\relax\relax \endgroup
} 

\usepackage[bookmarks=false]{hyperref}
\usepackage[all]{hypcap}

\title{The Cyclic-Routing UAV Problem is $\mathrm{PSPACE}$-Complete}
\author{Hsi-Ming Ho and Jo\"el Ouaknine}
\institute{Department of Computer Science, University of Oxford \\ Wolfson Building, Parks Road, Oxford, OX1 3QD, UK}
\begin{document}
\maketitle
\thispagestyle{plain}
\begin{abstract}
Consider a finite set of targets, with each target assigned a
\emph{relative deadline}, and each pair of targets assigned a fixed
transit \emph{flight time}. Given a flock of identical UAVs, can
one ensure that every target is repeatedly visited by some UAV at
intervals of duration at most the target's relative deadline? The
\mbox{\emph{\textbf{Cyclic-Routing UAV Problem}} \textsc{(cr-uav)}} is the question of
whether this task has a solution.

This problem can straightforwardly be solved in $\mathrm{PSPACE}$ by
modelling it as a network of timed automata. The special case of there
being a single UAV is claimed to be $\mathrm{NP}$-complete in the
literature. In this paper, we show that the \textsc{cr-uav} Problem is in
fact $\mathrm{PSPACE}$-complete even in the single-UAV case.
\end{abstract}

\section{Introduction}

Unmanned aerial vehicles (UAVs) have many uses, ranging from civilian
to military operations. Like other autonomous systems, they are
particularly well-suited to `dull, dirty, and/or dangerous'
missions~\cite{UAVS}.  A common scenario in such missions is that a
set of targets have to be visited by a limited number of UAVs.  This
has given rise to a large body of research on \emph{path planning} for
UAVs.\footnote{\url{http://scholar.google.com/} lists thousands of papers on the subject.} Depending on the specific application at hand, paths of UAVs
may be subject to various complex constraints, e.g., related to kinematics or
fuel (see, e.g.,~\cite{Alighanbari2003, Elizabeth2012, Yang2002,
  Richards2002}).

In this work, we consider the \emph{Cyclic-Routing UAV Problem} (\textsc{cr-uav})~\cite{Drucker2010}:
the decision version of a simple \emph{recurrent}
UAV path-planning problem in which each target must be visited not only
once but repeatedly, i.e., at intervals of prescribed maximal duration. 
Problems of this type have long been considered in many other fields
such as transportation~\cite{Orlin1982, Wollmer1990} and
robotics~\cite{Crama1997, Kats1997}.  More recently, a number of
game-theoretic frameworks have been developed to study similar
problems in the context of security~\cite{Tsai2009, Jain2010,
  Basilico2012}.

A special case of the problem (with a single UAV) is considered
in~\cite{Basilico2009, Basilico2012, Fargeas2013}, and is claimed to
be $\mathrm{NP}$-complete in~\cite{Basilico2012}.  However, the proof
of $\mathrm{NP}$-membership in~\cite{Basilico2012} is not
detailed.\footnote{A counterexample to a crucial claim
  in~\cite{Basilico2012} is given in Appendix~\ref{app:cex}.}  The
main result of the present paper is that the \textsc{cr-uav} Problem is
in fact $\mathrm{PSPACE}$-complete, even in the single-UAV case.
We note that this problem can be seen as a recurrent variant of the
decision version of the \emph{Travelling Salesman Problem with Time
  Windows} (\textsc{tsptw}) with upper bounds only (or \emph{TSP with Deadlines}~\cite{Bockenhauer2007}).  Its
$\mathrm{PSPACE}$-hardness hence stems from recurrence: the decision
version of the (non-recurrent) \textsc{tsptw} Problem is
$\mathrm{NP}$-complete~\cite{Savelsbergh1985}.


$\mathrm{PSPACE}$-membership of the (general) \textsc{cr-uav} Problem
follows straightforwardly by encoding the problem as the existence of
infinite paths in a network of timed automata; we briefly sketch the
argument in the next section. The bulk of the paper is then
devoted to establishing $\mathrm{PSPACE}$-hardness of the single-UAV
case. This is accomplished by reduction from the \textsc{periodic sat}
Problem, known to be $\mathrm{PSPACE}$-complete~\cite{Orlin1981}.

\section{Preliminaries}


\subsection{Scenario}

Let there be a set of targets and a number of identical UAVs.  Each
target has a \emph{\textbf{relative deadline}}: an upper bound
requirement on the time between successive visits by UAVs.  The UAVs
are allowed to fly freely between targets, with a \emph{\textbf{flight
    time}} given for each pair of targets: the amount of time
required for a UAV to fly from one of the targets to the other. We
assume that flight times are symmetric, that they obey the triangle
inequality, and that the flight time from target $v$ to target $v'$ is
zero iff $v$ and $v'$ denote the same target. In other words, flight
times are a metric on the set of targets. The goal is to decide
whether there is a way to coordinate UAVs such that no relative
deadline is ever violated. We make a few further assumptions:
\begin{itemize}
\item Initially, each UAV starts at some target; there may be more
  than one UAV at the same target.
\item The first visit to each target must take place at the latest by
  the expiration time of its relative deadline.
\item The UAVs are allowed to `wait' as long as they wish at any given
  target. 
\item Time units are chosen so that all relative deadlines and flight
  times are integers, and moreover all relative deadlines are
  interpreted as closed constraints (i.e., using non-strict inequalities).
\end{itemize}

\subsection{Modelling via Networks of Timed Automata}\label{subsec:membership}

We briefly sketch how to model the \textsc{cr-uav} Problem as the
existence of infinite non-Zeno paths in a network of B\"uchi timed
automata, following the notation and results of~\cite{Alur1998}, from
which $\mathrm{PSPACE}$-membership immediately follows.

Intuitively, one ascribes a particular timed automaton to each UAV and
to each target. Each UAV-automaton keeps track of the location of its
associated UAV, and enforces flight times by means of a single clock,
which is reset the instant the UAV leaves a given target. Each
target-automaton is likewise equipped with a single clock, keeping
track of time elapsed since the last visit by some UAV\@. The action
of a UAV visiting a target is modelled by synchronising on a
particular event; when this takes place, provided the target's
relative deadline has not been violated, the target resets its
internal clock and instantaneously visits a B\"uchi
location. Similarly, the action of a UAV leaving a target is modelled
by event synchronisation. Finally, since multiple UAVs may visit a
given target simultaneously, each target is in addition equipped with
a counter to keep track at any time of whether or not it is currently
being visited by some UAV.

The given instance of the \textsc{cr-uav} Problem therefore has a
solution iff there exists a non-Zeno run of the resulting network of
timed automata in which each B\"uchi accepting location is visited
infinitely often. By Thm.~$7$ of \cite{Alur1998}, this can be decided in
$\mathrm{PSPACE}$.

It is worth noting that, since all timing constraints are closed by
assumption, standard digitisation results apply (cf.~\cite{Henzinger1992}) and
it is sufficient to consider integer (i.e., discrete) time. In the
next section, we therefore present a discrete graph-based (and
timed-automaton independent) formulation of the problem specialised to
a single UAV, in order to establish $\mathrm{PSPACE}$-hardness.

\subsection{Weighted Graph Formulation}

The solution to a single-UAV instance of the \textsc{cr-uav} Problem
consists of an infinite path from target to target in which each
target is visited infinitely often, at time intervals never greater
than the target's relative deadline. One may clearly assume that the
UAV never `lingers' at any given target, i.e., targets are visited
instantaneously.  Formally, a single-UAV instance of the \textsc{cr-uav} Problem can be
described as follows. Let $V$ be a set of $n \geq 2$ vertices, with
each vertex $v \in V$ assigned a strictly positive integer weight
$\mathit{RD}(v)$ (intuitively, the relative deadline of target
$v$). Consider a weighted undirected clique over $V$, i.e., to each
pair of vertices $(v,v')$ with $v \neq v'$, one assigns a strictly
positive integer weight $\mathit{FT}(v,v')$ (intuitively, the flight
time from $v$ to $v'$). In addition we require that $\mathit{FT}$ be
symmetric and satisfy the triangle inequality.

%
%



Let $G = \langle V,\mathit{RD}, \mathit{FT}\rangle$ be an instance of
the above data. Given a finite path $u$ in (the clique associated
with) $G$, the \emph{\textbf{duration}} $\mathit{dur}(u)$ of $u$ is
defined to be the sum of the weights of the edges in $u$. A
\emph{\textbf{solution}} to $G$ is an infinite path $s$ through $G$ with the
following properties:
\begin{itemize}
\item $s$ visits every vertex in $V$ infinitely often;
\item Any finite subpath of $s$ that starts and ends at
  consecutive occurrences of a given vertex $v$ must have duration at
  most $\mathit{RD}(v)$.
\end{itemize}

\begin{definition}[The \textsc{cr-uav} Problem with a Single UAV]
Given $G$ as described above, does $G$ have a solution?
\end{definition}

As pointed out in~\cite{Fargeas2013}, if a
solution exists at all then a \emph{periodic} solution can be found, i.e., an
infinite path in which the targets are visited repeatedly in the same
order.


\subsection{The \textsc{periodic sat} Problem}

\textsc{periodic sat} is one of the many $\mathrm{PSPACE}$-complete
problems introduced in~\cite{Orlin1981}.  In the following definition
(and in the rest of this paper), let $\overline{x}$ be a finite set of
variables and let $\overline{x}^j$ be the set of variables obtained
from $\overline{x}$ by adding a superscript $j$ to each variable.

\begin{definition}[The \textsc{periodic sat} Problem~\cite{Orlin1981}]
Consider a CNF formula $\varphi(0)$ over $\overline{x}^0 \cup
\overline{x}^1$. Let $\varphi(j)$ be the formula obtained from
$\varphi(0)$ by replacing all variables $x_i^0 \in \overline{x}^0$ by
$x_i^j$ and all variables $x_i^1 \in \overline{x}^1$ by $x_i^{j+1}$.
Is there an assignment of $\bigcup_{j \geq 0} \overline{x}^j$ such
that $\bigwedge_{j \geq 0} \varphi(j)$ is satisfied?
\end{definition}



\section{$\mathrm{PSPACE}$-Hardness}

In this section, we give a reduction from the \textsc{periodic sat}
Problem to the \textsc{cr-uav} Problem with a single UAV\@.
Consider a CNF formula $\varphi(0) = c_1 \wedge \cdots \wedge c_h$
over $\overline{x}^0 = \{x_1^0, \ldots, x_m^0\}$ and $\overline{x}^1 =
\{x_1^1, \ldots, x_m^1\}$.  Without loss of generality, we assume that
each clause $c_j$ of $\varphi(0)$ is non-trivial (i.e., $c_j$ does not
contain both positive and negative occurrences of a variable) and $m >
2$, $h > 0$. We can construct an instance $G$ of the \textsc{cr-uav}
Problem (with the largest constant having magnitude $O(m^2h)$ and $|V|
= O(mh)$) such that $\bigwedge_{j \geq 0} \varphi(j)$ is satisfiable
if and only if $G$ has a solution.

The general idea of the reduction can be described as follows. We
construct \emph{variable gadgets} that can be traversed in two
`directions' (corresponding to assignments $\mathbf{true}$ and
$\mathbf{false}$ to variables). A \emph{clause vertex} is visited if
the corresponding clause is satisfied by the assignment.  Crucially,
we use \emph{consistency gadgets}, in which we set the relative
deadlines of the vertices carefully to ensure that the directions of
traversals of the variable gadgets for $\overline{x}^1$ (corresponding
to a particular assignment of variables) in a given iteration is
consistent with the directions of traversals of the variable gadgets
for $\overline{x}^0$ in the next iteration.

\subsection{The Construction}

We describe and explain each part of $G$ in detail.
The reader is advised to glance ahead to Figure~\ref{fig:stack} to form an
impression of $G$.  Note that for ease of presentation, we temporarily
relax the requirement that $\mathit{FT}$ be a metric and describe $G$ as an incomplete graph.\footnote{In
  the single-UAV case, if the $\mathit{FT}$ of some edge is greater
  than any value in $\mathit{RD}$, that edge can simply be seen as
  non-existent.}  In what follows, let $l = 24h + 34$ and
\[T = 2
\Big( m\big(2 (3m+1)l + l\big) + m\big(2 (3m+2)l + l\big) + l + 2h
\Big)\,.\]

\subsubsection{Variable Gadgets}
For each variable $x^0_i$, we construct (as a subgraph of $G$) a \emph{variable gadget}.
It consists of the following vertices (see Figure~\ref{fig:vargadget}):
\begin{itemize}
\item Three vertices on the left side ($\mathit{LS}_i = \{v^{t, L}_i, v^{m, L}_i, v^{b, L}_i\}$)
\item Three vertices on the right side ($\mathit{RS}_i = \{v^{t, R}_i, v^{m, R}_i, v^{b, R}_i\}$)
\item A `\emph{clause box}' ($\mathit{CB}_{i}^j  = \{ v^{a, j}_{i}, v^{b, j}_{i}, v^{c, j}_{i}, v^{d, j}_{i}, v^{e, j}_{i}, v^{f, j}_{i} \}$) for each $j \in \{1, \ldots, h\}$
\item A `\emph{separator box}' ($\mathit{SB}_{i}^j = \{ v^{\bar{a}, j}_{i}, v^{\bar{b}, j}_{i}, v^{\bar{c}, j}_{i}, v^{\bar{d}, j}_{i}, v^{\bar{e}, j}_{i}, v^{\bar{f}, j}_{i} \}$) for each $j \in \{0, \ldots, h\}$
\item A vertex at the top ($v_{top}$ if $i = 0$, $v_{i-1}$ otherwise)
\item A vertex at the bottom ($v_i$).
\end{itemize}

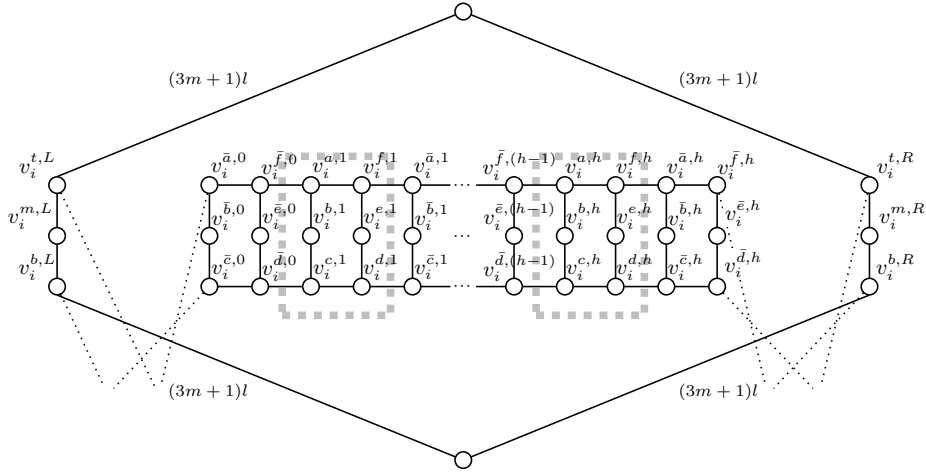
\begin{figure}[h]
\centering
\begin{tikzpicture}[-,>=stealth', auto, transform shape, node distance=2.5cm, scale=0.9,
                    semithick, every state/.style={fill=none,draw=black,text=black,shape=circle,scale=0.3}]

\node[state]               (vsrc)  {$$};

\node[state, draw=none]	   (vi) [below=3cm of vsrc]	{\LARGE $\boldsymbol{\cdots}$};
\node[state, draw=none]	   (vii) [above=.37cm of vi]	{\LARGE $\boldsymbol{\cdots}$};
\node[state, draw=none]	   (vio) [below=.37cm of vi]	{\LARGE $\boldsymbol{\cdots}$};

\node[state]			   (v5) [left of =vi] {$$};
\node[state]			   (v5i) [above=.5cm of v5] {$$};
\node[state]			   (v5o) [below=.5cm of v5] {$$};

\node[state]			   (v4) [left of =v5] {$$};
\node[state]			   (v4i) [below=.5cm of v4] {$$};
\node[state]			   (v4o) [above=.5cm of v4] {$$};

\node[state]			   (v3) [left of =v4] {$$};
\node[state]			   (v3i) [above=.5cm of v3] {$$};
\node[state]			   (v3o) [below=.5cm of v3] {$$};

\node[state]			   (v2) [left of =v3] {$$};
\node[state]			   (v2i) [below=.5cm of v2] {$$};
\node[state]			   (v2o) [above=.5cm of v2] {$$};

\node[state]			   (v1) [left of =v2] {$$};
\node[state]			   (v1i) [above=.5cm of v1] {$$};
\node[state]			   (v1o) [below=.5cm of v1] {$$};

\node[state]			   (v1') [left=2cm of v1] {$$};
\node[state]			   (v1'i) [above=.5cm of v1'] {$$};
\node[state]			   (v1'o) [below=.5cm of v1'] {$$};

\node[state]				(v6) [right of =vi] {$$};
\node[state]			   (v6i) [below=.5cm of v6] {$$};
\node[state]			   (v6o) [above=.5cm of v6] {$$};

\node[state]				(v7) [right of =v6] {$$};
\node[state]			   (v7i) [above=.5cm of v7] {$$};
\node[state]			   (v7o) [below=.5cm of v7] {$$};

\node[state]				(v8) [right of =v7] {$$};
\node[state]			   (v8i) [below=.5cm of v8] {$$};
\node[state]			   (v8o) [above=.5cm of v8] {$$};

\node[state]				(v9) [right of =v8] {$$};
\node[state]			   (v9i) [above=.5cm of v9] {$$};
\node[state]			   (v9o) [below=.5cm of v9] {$$};

\node[state]				(v10) [right of =v9] {$$};
\node[state]			   (v10i) [below=.5cm of v10] {$$};
\node[state]			   (v10o) [above=.5cm of v10] {$$};

\node[state]				(v10') [right=2cm of v10] {$$};
\node[state]			   (v10'i) [above=.5cm of v10'] {$$};
\node[state]			   (v10'o) [below=.5cm of v10'] {$$};

\node[state]               (vsnk) [below=3cm of vi]  {$$};

\node (box1) [draw=gray!50, fit=(v3i) (v4i), dashed, line width=1mm, inner sep=3mm] {};
\node (box2) [draw=gray!50, fit=(v7i) (v8i), dashed, line width=1mm, inner sep=3mm] {};

\node[state, draw=none]				 (p') [below=3cm of $(v10')!0.5!(v10)$]  {$$};

\node[state, draw=none]			   (idown') [above left=.6cm and .2cm of p'] {$$};

\node[state, draw=none]			   (oup') [above right=.6cm and .2cm of p'] {$$};

\node[state, draw=none]				 (p) [below=3cm of $(v1')!0.5!(v1)$]  {$$};

\node[state, draw=none]			   (idown) [above left=.6cm and .2cm of p] {$$};

\node[state, draw=none]			   (oup) [above right=.6cm and .2cm of p] {$$};

\path (vsrc) edge node [swap] {\scriptsize $(3m+1)l$} (v1'i.north)
		(vsrc) edge node {\scriptsize $(3m+1)l$} (v10'i.north)

        (v1'i) edge node [right=-1mm] {\scriptsize $$} (v1')
        (v1') edge node [right=-1mm] {\scriptsize $$} (v1'o)

        (v1i) edge node [right=-1mm] {\scriptsize $$} (v1)
        (v1) edge node [right=-1mm] {\scriptsize $$} (v1o)

        (v2i) edge node [right=-1mm] {\scriptsize $$} (v2)
        (v2) edge node [right=-1mm] {\scriptsize $$} (v2o)

        (v3i) edge node [right=-1mm] {\scriptsize $$} (v3)
        (v3) edge node [right=-1mm] {\scriptsize $$} (v3o)

        (v4i) edge node [right=-1mm] {\scriptsize $$} (v4)
        (v4) edge node [right=-1mm] {\scriptsize $$} (v4o)

        (v5i) edge node [right=-1mm] {\scriptsize $$} (v5)
        (v5) edge node [right=-1mm] {\scriptsize $$} (v5o)

        (v6i) edge node [right=-1mm] {\scriptsize $$} (v6)
        (v6) edge node [right=-1mm] {\scriptsize $$} (v6o)

        (v7i) edge node [right=-1mm] {\scriptsize $$} (v7)
        (v7) edge node [right=-1mm] {\scriptsize $$} (v7o)

        (v8i) edge node [right=-1mm] {\scriptsize $$} (v8)
        (v8) edge node [right=-1mm] {\scriptsize $$} (v8o)

        (v9i) edge node [right=-1mm] {\scriptsize $$} (v9)
        (v9) edge node [right=-1mm] {\scriptsize $$} (v9o)

        (v10i) edge node [right=-1mm] {\scriptsize $$} (v10)
        (v10) edge node [right=-1mm] {\scriptsize $$} (v10o)

        (v10'i) edge node [right=-1mm] {\scriptsize $$} (v10')
        (v10') edge node [right=-1mm] {\scriptsize $$} (v10'o)

        (v1i) edge node [below=-1mm] {\scriptsize $$} (v2o)
        (v2o) edge node [below=-1mm] {\scriptsize $$} (v3i)
        (v3i) edge node [below=-1mm] {\scriptsize $$} (v4o)
        (v4o) edge node [below=-1mm] {\scriptsize $$} (v5i)
        (v5i) edge node [below=-1mm] {\scriptsize $$} (vii)
        (vii) edge node [below=-1mm] {\scriptsize $$} (v6o)
        (v6o) edge node [below=-1mm] {\scriptsize $$} (v7i)
        (v7i) edge node [below=-1mm] {\scriptsize $$} (v8o)
        (v8o) edge node [below=-1mm] {\scriptsize $$} (v9i)
        (v9i) edge node [below=-1mm] {\scriptsize $$} (v10o)

        (v1o) edge node [swap, below=-1mm] {\scriptsize $$} (v2i)
        (v2i) edge node [swap, below=-1mm] {\scriptsize $$} (v3o)
        (v3o) edge node [swap, below=-1mm] {\scriptsize $$} (v4i)
        (v4i) edge node [swap, below=-1mm] {\scriptsize $$} (v5o)
        (v5o) edge node [swap, below=-1mm] {\scriptsize $$} (vio)
        (vio) edge node [swap, below=-1mm] {\scriptsize $$} (v6i)
        (v6i) edge node [swap, below=-1mm] {\scriptsize $$} (v7o)
        (v7o) edge node [swap, below=-1mm] {\scriptsize $$} (v8i)
        (v8i) edge node [swap, below=-1mm] {\scriptsize $$} (v9o)
        (v9o) edge node [swap, below=-1mm] {\scriptsize $$} (v10i)

        (v1'o.south) edge node [swap] {\scriptsize $(3m+1)l$} (vsnk)
        (v10'o.south) edge node {\scriptsize $(3m+1)l$} (vsnk)

		  (v1'o) edge [dotted] node {$$} (idown)
		  (v1o) edge [dotted] node {$$} (idown)

		  (v1'i) edge [dotted] node {$$} (oup)
		  (v1i) edge [dotted] node {$$} (oup)

		  (v10'o) edge [dotted] node {$$} (idown')
		  (v10o) edge [dotted] node {$$} (idown')

		  (v10'i) edge [dotted] node {$$} (oup')
		  (v10i) edge [dotted] node {$ $} (oup')

;

\node [xshift=-2mm, yshift=2mm] at (v1'i.north west) {$v^{t, L}_i$};
\node [xshift=-3mm, yshift=2mm] at (v1'.north west) {$v^{m, L}_i$};
\node [xshift=-2mm, yshift=2mm] at (v1'o.north west) {$v^{b, L}_i$};

\node [xshift=2mm, yshift=2mm] at (v1i.north east) {$v^{\bar{a}, 0}_{i}$};
\node [xshift=2mm, yshift=2mm] at (v1.north east)  {$v^{\bar{b}, 0}_{i}$};
\node [xshift=2mm, yshift=2mm] at (v1o.north east) {$v^{\bar{c}, 0}_{i}$};

\node [xshift=2mm, yshift=2mm] at (v2i.north east) {$v^{\bar{d}, 0}_{i}$};
\node [xshift=2mm, yshift=2mm] at (v2.north east)  {$v^{\bar{e}, 0}_{i}$};
\node [xshift=2mm, yshift=2mm] at (v2o.north east) {$v^{\bar{f}, 0}_{i}$};

\node [xshift=2mm, yshift=2mm] at (v3i.north east) {$v^{a, 1}_{i}$};
\node [xshift=2mm, yshift=2mm] at (v3.north east)  {$v^{b, 1}_{i}$};
\node [xshift=2mm, yshift=2mm] at (v3o.north east) {$v^{c, 1}_{i}$};

\node [xshift=2mm, yshift=2mm] at (v4i.north east) {$v^{d, 1}_{i}$};
\node [xshift=2mm, yshift=2mm] at (v4.north east)  {$v^{e, 1}_{i}$};
\node [xshift=2mm, yshift=2mm] at (v4o.north east) {$v^{f, 1}_{i}$};

\node [xshift=2mm, yshift=2mm] at (v5i.north east) {$v^{\bar{a}, 1}_{i}$};
\node [xshift=2mm, yshift=2mm] at (v5.north east)  {$v^{\bar{b}, 1}_{i}$};
\node [xshift=2mm, yshift=2mm] at (v5o.north east) {$v^{\bar{c}, 1}_{i}$};

\node [xshift=0mm, yshift=2mm] at (v6i.north east) {$v^{\bar{d}, (h-1)}_{i}$};
\node [xshift=0mm, yshift=2mm] at (v6.north east)  {$v^{\bar{e}, (h-1)}_{i}$};
\node [xshift=0mm, yshift=2mm] at (v6o.north east) {$v^{\bar{f}, (h-1)}_{i}$};

\node [xshift=2mm, yshift=2mm] at (v7i.north east) {$v^{a, h}_{i}$};
\node [xshift=2mm, yshift=2mm] at (v7.north east)  {$v^{b, h}_{i}$};
\node [xshift=2mm, yshift=2mm] at (v7o.north east) {$v^{c, h}_{i}$};

\node [xshift=2mm, yshift=2mm] at (v8i.north east) {$v^{d, h}_{i}$};
\node [xshift=2mm, yshift=2mm] at (v8.north east)  {$v^{e, h}_{i}$};
\node [xshift=2mm, yshift=2mm] at (v8o.north east) {$v^{f, h}_{i}$};

\node [xshift=2mm, yshift=2mm] at (v9i.north east) {$v^{\bar{a}, h}_{i}$};
\node [xshift=2mm, yshift=2mm] at (v9.north east)  {$v^{\bar{b}, h}_{i}$};
\node [xshift=2mm, yshift=2mm] at (v9o.north east) {$v^{\bar{c}, h}_{i}$};
                                                                    
\node [xshift=3mm, yshift=2.5mm, fill=white, inner sep=0.5pt] at (v10i.north east) {$v^{\bar{d}, h}_i$};
\node [xshift=3mm, yshift=2.5mm, fill=white, inner sep=0.5pt] at (v10.north east)  {$v^{\bar{e}, h}_i$};
\node [xshift=2mm, yshift=2mm] at (v10o.north east) 										  {$v^{\bar{f}, h}_i$};

\node [xshift=3mm, yshift=2mm] at (v10'i.north east) {$v^{t, R}_i$};
\node [xshift=4mm, yshift=2mm] at (v10'.north east)  {$v^{m, R}_i$};
\node [xshift=3mm, yshift=2mm] at (v10'o.north east) {$v^{b, R}_i$};

\end{tikzpicture}
\caption{The variable gadget for $x^0_i$}
\label{fig:vargadget}
\end{figure}

The clause boxes for $j \in \{1, \ldots, h\}$ are aligned horizontally in the figure.
A separator box is laid between each adjacent pair of clause boxes and at both ends.
This row of boxes (\(\mathit{Row}_i = \bigcup_{j \in \{1, \ldots, h\}} \mathit{CB}_i^j \cup \bigcup_{j \in \{0, \ldots, h\}} \mathit{SB}_i^j\))
is then put between $\mathit{LS}_i$ and $\mathit{RS}_i$.
The $\mathit{RD}$ of all vertices
$v \in \mathit{LS}_i \cup \mathit{RS}_i \cup \mathit{Row}_i$ are set to $T + l + 2h$.

The vertices are connected as indicated by solid lines in the figure.
The four `long' edges in the figure have their $\mathit{FT}$ set to $(3m+1)l$
while all other edges have $\mathit{FT}$ equal to $2$, e.g.,
$\mathit{FT}(v_{top}, v^{t, L}_1)  = (3m + 1)l$ and $\mathit{FT}(v^{b, 1}_1, v^{c, 1}_1) = 2$.
There is an exception though: $\mathit{FT}(v_m^{b, L}, v_m)$ and $\mathit{FT}(v_m^{b, R}, v_m)$
(in the variable gadget for $x^0_m$) are equal to $(3m + 2)l$.

The variable gadgets for variables $x^1_i$ are constructed almost identically.
The three vertices on the left and right side are now $\mathit{LS}_{i+m}$ and $\mathit{RS}_{i+m}$.
The set of vertices in the row is now
\(\mathit{Row}_{i + m} = \bigcup_{j \in \{1, \ldots, h\}} \mathit{CB}_{i + m}^j \cup \bigcup_{j \in \{0, \ldots, h\}} \mathit{SB}_{i + m}^j\).
The vertex at the top is $v_{i + m - 1}$ and the vertex at the bottom is $v_{i + m}$ ($i \neq m$)
or $v_{bot}$ ($i = m$). The $\mathit{RD}$ of vertices in
$\mathit{LS}_{i + m} \cup \mathit{RS}_{i + m} \cup \mathit{Row}_{i + m}$ are set to $T + l + 2h$,
and the $\mathit{FT}$ of the edges are set as before, except that all
the `long' edges now have $\mathit{FT}$ equal to $(3m+2)l$.

Now consider the following ordering of variables:
\[
x_1^0, x_2^0, \ldots, x_m^0, x_1^1, x_2^1, \ldots, x_m^1 \,.
\]
Observe that the variable gadgets for two `neighbouring' variables (with respect to this ordering) have a vertex in common.
To be precise, the set of shared vertices is $S = \{v_1, \ldots, v_{2m-1}\}$.
We set the $\mathit{RD}$ of all vertices in $S$ to $T + 2h$ and the $\mathit{RD}$ of
$v_{top}$ and $v_{bot}$ to $T$. 

\subsubsection{Clause Vertices}

For each clause $c_j$ in $\varphi(0)$, there is a \emph{clause vertex} $v^{c_j}$
with $\mathit{RD}$ set to $\frac{3}{2}T$.
If $x_i^0$ occurs in $c_j$ as a literal, we connect the $j$-th clause box in the variable gadget for $x_i^0$
to $v^{c_j}$ as shown in Figure~\ref{fig:detour} and set the $\mathit{FT}$ of these new edges to $2$
(e.g., $\mathit{FT}(v^{c_j}, v_i^{c, j}) = \mathit{FT}(v^{c_j}, v_i^{d, j}) = 2$). 
If instead $\neg x_i^0$ occurs in $c_j$, then $v^{c_j}$ is connected to $v_i^{a, j}$ and $v_i^{f, j}$
(with $\mathit{FT}$ equal to $2$). Likewise, the variable gadget for $x_i^1$ may be connected to
$v^{c_j}$ via $\{v_{i+m}^{c, j}, v_{i+m}^{d, j}\}$ (if $x_i^1$ occurs in $c_j$)
or $\{v_{i+m}^{a, j}, v_{i+m}^{f, j}\}$ (if $\neg x_i^1$ occurs in $c_j$).

\begin{figure}[h]
\begin{minipage}[b]{0.5\linewidth}
\centering

\begin{tikzpicture}[-,>=stealth', auto, transform shape, node distance=2.5cm,
                    semithick, every state/.style={fill=none,draw=black,text=black,shape=circle,scale=0.3}]

\node[state, draw=none]	   (vi) 								{\LARGE $\boldsymbol{\cdots}$};
\node[state, draw=none]	   (vii) [above=.37cm of vi]	{\LARGE $\boldsymbol{\cdots}$};
\node[state, draw=none]	   (vio) [below=.37cm of vi]	{\LARGE $\boldsymbol{\cdots}$};

\node[state]			   (v5) [left of =vi] {$$};
\node[state]			   (v5i) [above=.5cm of v5] {$$};
\node[state]			   (v5o) [below=.5cm of v5] {$$};

\node[state]			   (v4) [left of =v5] {$$};
\node[state]			   (v4i) [below=.5cm of v4] {$$};
\node[state]			   (v4o) [above=.5cm of v4] {$$};

\node[state, draw=none]	   (vi') [left of =v4]				{\LARGE $\boldsymbol{\cdots}$};
\node[state, draw=none]	   (vii') [above=.37cm of vi']	{\LARGE $\boldsymbol{\cdots}$};
\node[state, draw=none]	   (vio') [below=.37cm of vi']	{\LARGE $\boldsymbol{\cdots}$};

\node[state]				 (cj) [below=2cm of $(v4i)!0.5!(v5o)$]  {$$};

\node (box1) [draw=gray!50, fit=(v4o) (v5o), dashed, line width=1mm, inner sep=3mm] {};
\node [xshift=1.5mm, yshift=-1.5mm] at (cj.south east) {$v^{c_j}$};

\path   (vii') edge [->, line width=1mm] node {\scriptsize $2$} (v4o)
        (vio') edge node [swap] {\scriptsize $2$} (v4i)
		  (v5i) edge [->, line width=1mm] node {\scriptsize $2$} (vii)
        (vio) edge node {\scriptsize $2$} (v5o)
			
        (v4) edge [->, line width=1mm] node [swap] {\scriptsize $2$} (v4i)
        (v4o) edge [->, line width=1mm] node [swap] {\scriptsize $2$} (v4)

        (v4i) edge node [swap] {\scriptsize $2$} (v5o)
        (v5i) edge node [swap] {\scriptsize $2$} (v4o)

        (v4i) edge [->, line width=1mm] node [swap] {\scriptsize $2$} (cj.north west)
        (cj.north east) edge [->, line width=1mm] node [swap] {\scriptsize $2$} (v5o)

        (v5) edge [->, line width=1mm] node {\scriptsize $2$} (v5i)
        (v5o) edge [->, line width=1mm] node {\scriptsize $2$} (v5);

\end{tikzpicture}
\caption{The variable occurs positively in $c_j$}
\label{fig:detour}
\end{minipage}
\begin{minipage}[b]{0.5\linewidth}
\centering
\begin{tikzpicture}[-,>=stealth', auto, transform shape, node distance=7cm, scale=0.9,
                    semithick, every state/.style={fill=none,draw=black,text=black,shape=circle,scale=0.3}]

\node[state]			   (p) {$$};


\node[state]			   (idown) [above left of =p] {$$};

\node[state]			   (oup) [above right of =p] {$$};

\node[state]			   (iup) [below left of =p] {$$};

\node[state]			   (odown) [below right of =p] {$$};

\node [xshift=0mm, yshift=4mm] at (p.north) {$\mathit{pvt}_i^L$};

\node [xshift=2mm, yshift=2mm] at (idown.north east) {${in}_i^{\downarrow, L}$};
\node [xshift=2mm, yshift=2mm] at (oup.north east) {${out}_i^{\uparrow, L}$};
\node [xshift=-2mm, yshift=2mm] at (iup.north west) {${in}_i^{\uparrow, L}$};
\node [xshift=2mm, yshift=2mm] at (odown.north east) {${out}_i^{\downarrow, L}$};

\path   (p) edge [swap] node {\scriptsize $2$} (oup)
		  (idown) edge node {\scriptsize $2$} (p)
		  (p) edge [swap] node {\scriptsize $2$} (odown)
		  (iup) edge node {\scriptsize $2$} (p);
\end{tikzpicture}
\caption{A consistency gadget $\mathit{LCG}_i$}
\label{fig:aux}
\end{minipage}
\end{figure}

\vspace{-0.5cm}

\subsubsection{Consistency Gadgets}
For each $i \in \{1, \ldots, m \}$, we construct two \emph{consistency gadgets}
$\mathit{LCG}_i$ (see Figure~\ref{fig:aux}) and $\mathit{RCG}_i$.
In $\mathit{LCG}_i$, the vertex at the centre ($\mathit{pvt}_i^{t, L}$)
has $\mathit{RD}$ equal to $\frac{1}{2}T + m\big(2(3m+2)l + l\big) - (2i - 1)l + 4h$.
The other four vertices ($\mathit{in}^{\downarrow, L}_i$, $\mathit{out}^{\uparrow, L}_i$, $\mathit{in}^{\uparrow, L}_i$ and $\mathit{out}^{\downarrow, L}_i$)
have $\mathit{RD}$ equal to $\frac{3}{2}T$. The $\mathit{FT}$ from $\mathit{pvt}_i^{t, L}$ to any
of the other four vertices is $2$. $\mathit{RCG}_i$ is identical except that the subscripts on
the vertices change from $L$ to $R$.

$\mathit{LCG}_i$ and $\mathit{RCG}_i$ are connected to the variable gadgets for $x_i^0$ and $x_i^1$
as in Figure~\ref{fig:connectaux}.
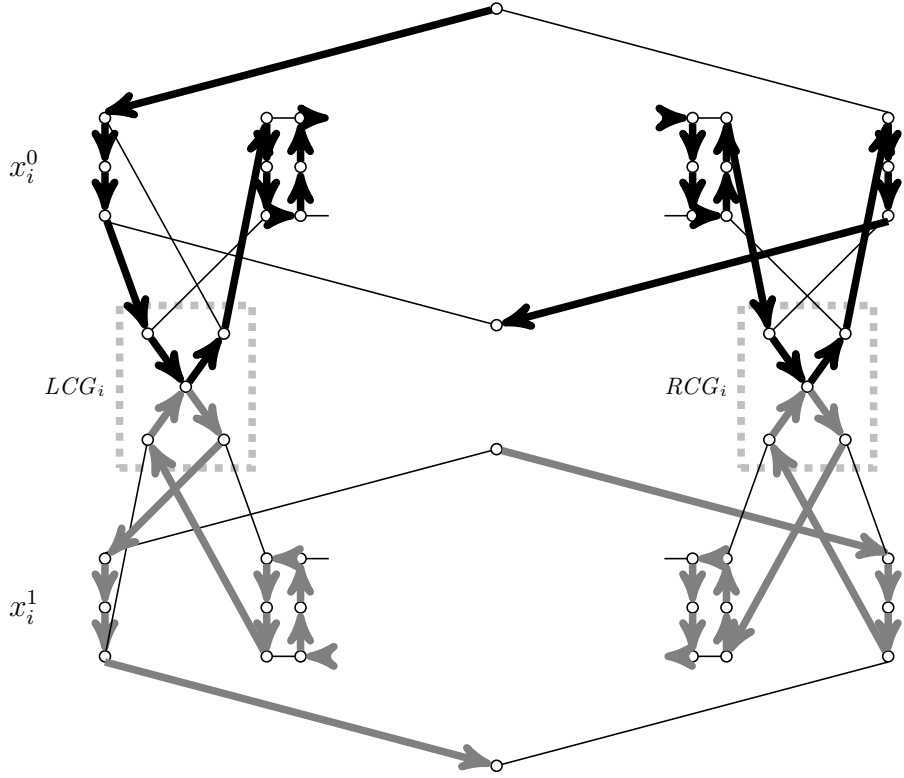
\begin{figure}[h]
\centering

\begin{tikzpicture}[-,>=stealth', auto, transform shape, node distance=2.5cm,
                    semithick, every state/.style={fill=none,draw=black,text=black,shape=circle,scale=0.18}]

\node[state, draw=none]			   (v3) {$$};
\node[state, draw=none]			   (v3i) [above=.5cm of v3] {$$};
\node[state, draw=none]			   (v3o) [below=.5cm of v3] {$$};

\node[state]			   (v2) [left of =v3] {$$};
\node[state]			   (v2i) [below=.5cm of v2] {$$};
\node[state]			   (v2o) [above=.5cm of v2] {$$};

\node[state]			   (v1) [left of =v2] {$$};
\node[state]			   (v1i) [above=.5cm of v1] {$$};
\node[state]			   (v1o) [below=.5cm of v1] {$$};

\node[state]			   (v0) [left=2cm of v1] {$$};
\node[state]			   (v0i) [above=.5cm of v0] {$$};
\node[state]			   (v0o) [below=.5cm of v0] {$$};

\node[state]				(vt1)	[above right=2cm and 2.5cm of v2] {$$};
\node[state]				(vb1)	[below right=2cm and 2.5cm of v2] {$$};

\node[state]				(vt2)	[below=1.5cm of vb1] {$$};

\node[state]			   (v5) [below right=2cm and 2.5cm of vt1] {$$};
\node[state]			   (v5i) [below=.5cm of v5] {$$};
\node[state]			   (v5o) [above=.5cm of v5] {$$};

\node[state, draw=none]			   (v4) [left of =v5] {$$};
\node[state, draw=none]			   (v4i) [above=.5cm of v4] {$$};
\node[state, draw=none]			   (v4o) [below=.5cm of v4] {$$};

\node[state]			   (v6) [right of =v5] {$$};
\node[state]			   (v6i) [above=.5cm of v6] {$$};
\node[state]			   (v6o) [below=.5cm of v6] {$$};

\node[state]			   (v7) [right=2cm of v6] {$$};
\node[state]			   (v7i) [above=.5cm of v7] {$$};
\node[state]			   (v7o) [below=.5cm of v7] {$$};

\node[state]			   (v2') [below left=2cm and 2.5cm of vt2] {$$};
\node[state]			   (v2i') [below=.5cm of v2'] {$$};
\node[state]			   (v2o') [above=.5cm of v2'] {$$};

\node[state, draw=none]			   (v3') [right of =v2'] {$$};
\node[state, draw=none]			   (v3i') [above=.5cm of v3'] {$$};
\node[state, draw=none]			   (v3o') [below=.5cm of v3'] {$$};

\node[state]			   (v1') [left of =v2'] {$$};
\node[state]			   (v1i') [above=.5cm of v1'] {$$};
\node[state]			   (v1o') [below=.5cm of v1'] {$$};

\node[state]			   (v0') [left=2cm of v1'] {$$};
\node[state]			   (v0i') [above=.5cm of v0'] {$$};
\node[state]			   (v0o') [below=.5cm of v0'] {$$};

\node[state]				(vb1')	[below right=2cm and 2.5cm of v2'] {$$};

\node[state]			   (v5') [below right=2cm and 2.5cm of vt2] {$$};
\node[state]			   (v5i') [below=.5cm of v5'] {$$};
\node[state]			   (v5o') [above=.5cm of v5'] {$$};

\node[state, draw=none]			   (v4') [left of =v5'] {$$};
\node[state, draw=none]			   (v4i') [above=.5cm of v4'] {$$};
\node[state, draw=none]			   (v4o') [below=.5cm of v4'] {$$};

\node[state]			   (v6') [right of =v5'] {$$};
\node[state]			   (v6i') [above=.5cm of v6'] {$$};
\node[state]			   (v6o') [below=.5cm of v6'] {$$};

\node[state]			   (v7') [right=2cm of v6'] {$$};
\node[state]			   (v7i') [above=.5cm of v7'] {$$};
\node[state]			   (v7o') [below=.5cm of v7'] {$$};

\node[state]				 (p') [below=2.85cm of $(v6)!0.5!(v7)$]  {$$};

\node[state]			   (idown') [above left=.6cm and .4cm of p'] {$$};

\node[state]			   (oup') [above right=.6cm and .4cm of p'] {$$};

\node[state]			   (iup') [below left=.6cm and .4cm  of p'] {$$};

\node[state]			   (odown') [below right=.6cm and .4cm of p'] {$$};

\node[state]				 (p) [below=2.85cm of $(v0)!0.5!(v1)$]  {$$};

\node[state]			   (idown) [above left=.6cm and .4cm of p] {$$};

\node[state]			   (oup) [above right=.6cm and .4cm of p] {$$};

\node[state]			   (iup) [below left=.6cm and .4cm of p] {$$};

\node[state]			   (odown) [below right=.6cm and .4cm of p] {$$};

\node [xshift=-10mm] at (v0.west) {\large $x_i^0$};

\node [xshift=-10mm] at (v0'.west) {\large $x_i^1$};

\node (box1) [draw=gray!50, fit=(idown) (odown), dashed, line width=1mm, inner sep=3mm] {};
\node (box2) [draw=gray!50, fit=(idown') (odown'), dashed, line width=1mm, inner sep=3mm] {};
\node [xshift=-14mm] at (p.west) {$\mathit{LCG}_i$};
\node [xshift=-14mm] at (p'.west) {$\mathit{RCG}_i$};

\path   (vt1) edge [->, line width=1mm] node {$$} (v0i.north)
	     (v0i) edge [->, line width=1mm] node {$$} (v0)
	     (v0) edge [->, line width=1mm] node {$$} (v0o)

	     (v1i) edge [->, line width=1mm] node {$$} (v1)
	     (v1) edge [->, line width=1mm] node {$$} (v1o)

	     (v2i) edge [->, line width=1mm] node {$$} (v2)
	     (v2) edge [->, line width=1mm] node {$$} (v2o)

	     (v1i) edge node {$$} (v2o)
	     (v1o) edge [->, line width=1mm] node {$$} (v2i)

	     (v2i) edge node {$$} (v3o)
	     (v2o) edge [->, line width=1mm] node {$$} (v3i)
	     (v0o.south) edge node {$$} (vb1)

		  (vt2) edge node {$$} (v0i'.north)
	     (v0i') edge [->, draw=black!50, line width=1mm] node {$$} (v0')
	     (v0') edge [->, draw=black!50, line width=1mm] node {$$} (v0o')

	     (v1i') edge [->, draw=black!50, line width=1mm] node {$$} (v1')
	     (v1') edge [->, draw=black!50, line width=1mm] node {$$} (v1o')

	     (v2i') edge [->, draw=black!50, line width=1mm] node {$$} (v2')
	     (v2') edge [->, draw=black!50, line width=1mm] node {$$} (v2o')

	     (v2o') edge [->, draw=black!50, line width=1mm] node {$$} (v1i')
	     (v1o') edge node {$$} (v2i')

	     (v3o') edge [->, draw=black!50, line width=1mm] node {$$} (v2i')
	     (v2o') edge node {$$} (v3i')
	     (v0o'.south) edge [->, draw=black!50, line width=1mm] node {$$} (vb1')

		  (vt1) edge node {$$} (v7i.north)
	     (v7i) edge [->, line width=1mm] node {$$} (v7)
	     (v7) edge [->, line width=1mm] node {$$} (v7o)

	     (v6) edge [->, line width=1mm] node {$$} (v6i)
	     (v6o) edge [->, line width=1mm] node {$$} (v6)

	     (v5) edge [->, line width=1mm] node {$$} (v5i)
	     (v5o) edge [->, line width=1mm] node {$$} (v5)

	     (v4i) edge [->, line width=1mm] node {$$} (v5o)
	     (v4o) edge node {$$} (v5i)

	     (v5i) edge [->, line width=1mm] node {$$} (v6o)
	     (v5o) edge node {$$} (v6i)
	     (v7o.south) edge  [->, line width=1mm] node {$$} (vb1)

		  (vt2) edge [->, draw=black!50, line width=1mm] node {$$} (v7i'.north)
	     (v7i') edge [->, draw=black!50, line width=1mm] node {$$} (v7')
	     (v7') edge [->, draw=black!50, line width=1mm] node {$$} (v7o')

	     (v6o') edge [->, draw=black!50, line width=1mm] node {$$} (v6')
	     (v6') edge [->, draw=black!50, line width=1mm] node {$$} (v6i')

	     (v5') edge [->, draw=black!50, line width=1mm] node {$$} (v5i')
	     (v5o') edge [->, draw=black!50, line width=1mm] node {$$} (v5')

	     (v4i') edge node {$$} (v5o')
	     (v5i') edge [->, draw=black!50, line width=1mm] node {$$} (v4o')

	     (v5i') edge node {$$} (v6o')
	     (v6i') edge [->, draw=black!50, line width=1mm] node {$$} (v5o')
	     (v7o'.south) edge node {$$} (vb1')

		  (p) edge [->, line width=1mm] node {$$} (oup)
		  (idown) edge [->, line width=1mm] node {$$} (p)
		  (p) edge [->, draw=black!50, line width=1mm] node {$$} (odown)
		  (iup) edge [->, draw=black!50, line width=1mm] node {$$} (p)

	     (v0o) edge [->, line width=1mm] node {$$} (idown)
	     (v1o) edge node {$$} (idown)

	     (v0i) edge node {$$} (oup)
	     (oup) edge [->, line width=1mm] node {$$} (v1i)

	     (v0o') edge node {$$} (iup)
	     (v1o') edge [->, draw=black!50, line width=1mm] node {$$} (iup)

	     (odown) edge [->, draw=black!50, line width=1mm] node {$$} (v0i')
	     (v1i') edge node {$$} (odown)

		  (p') edge [->, line width=1mm] node {$$} (oup')
		  (idown') edge [->, line width=1mm] node {$$} (p')
		  (p') edge [->, draw=black!50, line width=1mm] node {$$} (odown')
		  (iup') edge [->, draw=black!50, line width=1mm] node {$$} (p')

	     (v6i) edge [->, line width=1mm] node {$$} (idown')
		  (v7o) edge node {$$} (idown')
		  (v6o) edge node {$$} (oup')
	     (oup') edge [->, line width=1mm] node {$$} (v7i)

	     (v6i') edge node {$$} (iup')
		  (v7o') edge [->, draw=black!50, line width=1mm] node {$$} (iup')
		  (odown') edge [->, draw=black!50, line width=1mm] node {$$} (v6o')
	     (v7i') edge node {$ $} (odown');

\end{tikzpicture}
\caption{Connecting the variable gadgets for $x_i^0$ and $x_i^1$ to $\mathit{LCG}_i$ and $\mathit{RCG}_i$}
\label{fig:connectaux}
\end{figure}
The vertices ${in}_i^{\downarrow, L}, {out}_i^{\uparrow, L}, {in}_i^{\downarrow, R}, {out}_i^{\uparrow, R}$
are connected to certain vertices in the variable gadget for $x_i^0$---this allows $\mathit{pvt}_i^L$ and $\mathit{pvt}_i^R$ to be traversed `from above'. Similarly,
the edges connected to ${in}_i^{\uparrow, L}, {out}_i^{\downarrow, L}, {in}_i^{\uparrow, L}, {out}_i^{\downarrow, L}$ allow $\mathit{pvt}_i^L$ and $\mathit{pvt}_i^R$ to be traversed `from below'.
Formally, $\mathit{FT}(v, v') = 2$ if
\begin{itemize}
\item $v = \mathit{in}^{\downarrow, L}_i, v' \in \{ v_i^{b, L}, v_i^{\bar{c}, 0} \} \text{ or } v = \mathit{in}^{\downarrow, R}_i, v' \in \{ v_i^{\bar{f}, h}, v_i^{b, R} \}$
\item $v = \mathit{out}^{\uparrow, L}_i, v' \in \{ v_i^{t, L}, v_i^{\bar{a}, 0} \} \text{ or } v = \mathit{out}^{\uparrow, R}_i, v' \in \{ v_i^{\bar{d}, h}, v_i^{t, R} \}$
\item $v = \mathit{in}^{\uparrow, L}_i, v' \in \{ v_{(i+m)}^{b, L}, v_{(i+m)}^{\bar{c}, 0} \} \text{ or } v = \mathit{in}^{\uparrow, R}_i, v' \in \{ v_{(i + m)}^{\bar{f}, h}, v_{(i + m)}^{b, R} \}$
\item $v = \mathit{out}^{\downarrow, L}_i, v' \in \{ v_{(i+m)}^{t, L}, v_{(i+m)}^{\bar{a}, 0} \} \text{ or } v = \mathit{out}^{\downarrow, R}_i, v' \in \{ v_{(i+m)}^{\bar{d}, h}, v_{(i+m)}^{t, R} \}$.
\end{itemize}
Two parts of an intended path, which we will explain in more detail later, is also illustrated in Figure~\ref{fig:connectaux}.

Finally, there is a vertex $v_{mid}$ with $\mathit{RD}(v_{mid}) = T$ connected to $v_{bot}$ and $v_{top}$
with two edges, both with $\mathit{FT}$ equal to $\frac{1}{4}T$.
The $\mathit{FT}$ of all the missing edges are $2T$
(note that the largest value in $\mathit{RD}$ is less than $2T$, so these edges can never be taken).
This completes the construction of $G$. An example with $m = 3$
is given in Figure~\ref{fig:stack}, where vertices in $S$ (shared by two variable gadgets) are depicted as solid circles.

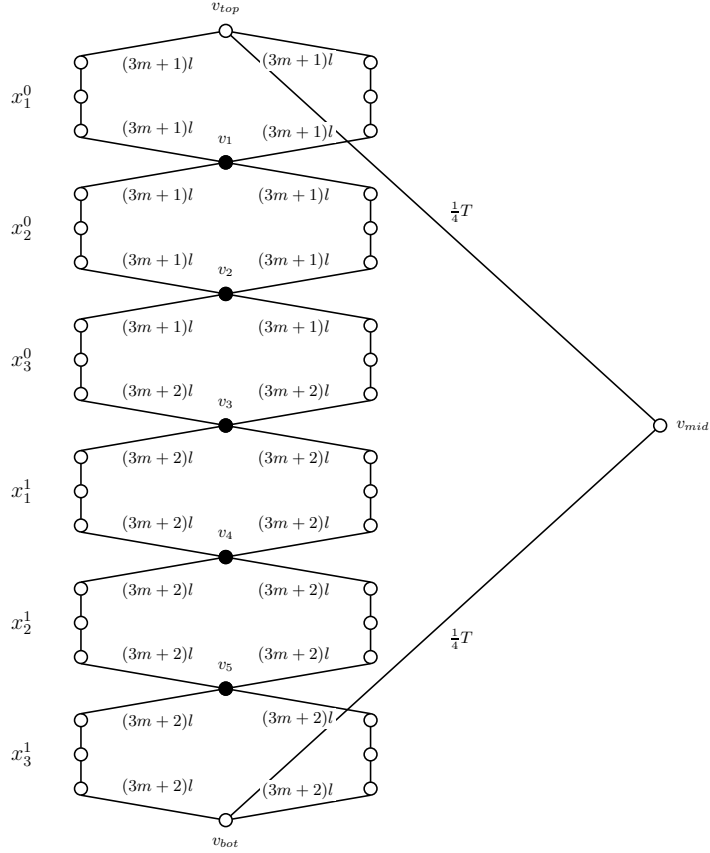
\begin{figure}[h]
\centering
\begin{tikzpicture}[-,>=stealth', auto, transform shape, node distance=.6cm, scale=0.7,
                    semithick, every state/.style={fill=none,draw=black,text=black,shape=circle, scale=0.3}]

\node[state, draw=none]			   (v3) {$$};
\node[state, draw=none]			   (v3i) [above=.25cm of v3] {$$};
\node[state, draw=none]			   (v3o) [below=.25cm of v3] {$$};

\node[state]			   (v0) [left=2.5cm of v3] {$$};
\node[state]			   (v0i) [above=.4cm of v0] {$$};
\node[state]			   (v0o) [below=.4cm of v0] {$$};

\node[state]				(vt1)	[above=1cm of v3] {$$};
\node[state]				(vb1)	[below=1cm of v3] {$$};

\node[state]			   (v7) [right=2.5cm of v3] {$$};
\node[state]			   (v7i) [above=.4cm of v7] {$$};
\node[state]			   (v7o) [below=.4cm of v7] {$$};

\node[state, draw=none]			   (v32) [below=1cm of vb1] {$$};
\node[state, draw=none]			   (v3i2) [above=.25cm of v32] {$$};
\node[state, draw=none]			   (v3o2) [below=.25cm of v32] {$$};

\node[state]			   (v02) [left=2.5cm of v32] {$$};
\node[state]			   (v0i2) [above=.4cm of v02] {$$};
\node[state]			   (v0o2) [below=.4cm of v02] {$$};

\node[state, fill=black]				(vt2)	[above=1cm of v32] {$$};
\node[state]				(vb2)	[below=1cm of v32] {$$};

\node[state]			   (v72) [right=2.5cm of v32] {$$};
\node[state]			   (v7i2) [above=.4cm of v72] {$$};
\node[state]			   (v7o2) [below=.4cm of v72] {$$};

\node[state, draw=none]			   (v33) [below=1cm of vb2] {$$};
\node[state, draw=none]			   (v3i3) [above=.25cm of v33] {$$};
\node[state, draw=none]			   (v3o3) [below=.25cm of v33] {$$};

\node[state]			   (v03) [left=2.5cm of v33] {$$};
\node[state]			   (v0i3) [above=.4cm of v03] {$$};
\node[state]			   (v0o3) [below=.4cm of v03] {$$};

\node[state, fill=black]				(vt3)	[above=1cm of v33] {$$};
\node[state]				(vb3)	[below=1cm of v33] {$$};

\node[state]			   (v73) [right=2.5cm of v33] {$$};
\node[state]			   (v7i3) [above=.4cm of v73] {$$};
\node[state]			   (v7o3) [below=.4cm of v73] {$$};

\node[state, draw=none]			   (v34) [below=1cm of vb3] {$$};
\node[state, draw=none]			   (v3i4) [above=.25cm of v34] {$$};
\node[state, draw=none]			   (v3o4) [below=.25cm of v34] {$$};

\node[state]			   (v04) [left=2.5cm of v34] {$$};
\node[state]			   (v0i4) [above=.4cm of v04] {$$};
\node[state]			   (v0o4) [below=.4cm of v04] {$$};

\node[state, fill=black]				(vt4)	[above=1cm of v34] {$$};
\node[state]				(vb4)	[below=1cm of v34] {$$};

\node[state]			   (v74) [right=2.5cm of v34] {$$};
\node[state]			   (v7i4) [above=.4cm of v74] {$$};
\node[state]			   (v7o4) [below=.4cm of v74] {$$};

\node[state]				(vmid)	[right=8cm of vt4] {$$};

\node[state, draw=none]			   (v35) [below=1cm of vb4] {$$};
\node[state, draw=none]			   (v3i5) [above=.25cm of v35] {$$};
\node[state, draw=none]			   (v3o5) [below=.25cm of v35] {$$};

\node[state]			   (v05) [left=2.5cm of v35] {$$};
\node[state]			   (v0i5) [above=.4cm of v05] {$$};
\node[state]			   (v0o5) [below=.4cm of v05] {$$};

\node[state, fill=black]				(vt5)	[above=1cm of v35] {$$};
\node[state]				(vb5)	[below=1cm of v35] {$$};

\node[state]			   (v75) [right=2.5cm of v35] {$$};
\node[state]			   (v7i5) [above=.4cm of v75] {$$};
\node[state]			   (v7o5) [below=.4cm of v75] {$$};

\node[state, draw=none]			   (v36) [below=1cm of vb5] {$$};
\node[state, draw=none]			   (v3i6) [above=.25cm of v36] {$$};
\node[state, draw=none]			   (v3o6) [below=.25cm of v36] {$$};

\node[state]			   (v06) [left=2.5cm of v36] {$$};
\node[state]			   (v0i6) [above=.4cm of v06] {$$};
\node[state]			   (v0o6) [below=.4cm of v06] {$$};

\node[state, fill=black]				(vt6)	[above=1cm of v36] {$$};
\node[state]				(vb6)	[below=1cm of v36] {$$};

\node[state]			   (v76) [right=2.5cm of v36] {$$};
\node[state]			   (v7i6) [above=.4cm of v76] {$$};
\node[state]			   (v7o6) [below=.4cm of v76] {$$};

\node [xshift=5mm] at (vmid.east) {\small $v_{mid}$};
\node [yshift=3mm] at (vt1.north) {\small $v_{top}$};
\node [yshift=-3mm] at (vb6.south) {\small $v_{bot}$};
\node [yshift=3mm] at (vt2.north) {\small $v_1$};
\node [yshift=3mm] at (vt3.north) {\small $v_2$};
\node [yshift=3mm] at (vt4.north) {\small $v_3$};
\node [yshift=3mm] at (vt5.north) {\small $v_4$};
\node [yshift=3mm] at (vt6.north) {\small $v_5$};

\node [xshift=-10mm] at (v0.west) {\large $x_1^0$};
\node [xshift=-10mm] at (v02.west) {\large $x_2^0$};
\node [xshift=-10mm] at (v03.west) {\large $x_3^0$};
\node [xshift=-10mm] at (v04.west) {\large $x_1^1$};
\node [xshift=-10mm] at (v05.west) {\large $x_2^1$};
\node [xshift=-10mm] at (v06.west) {\large $x_3^1$};

\path   (vmid) edge node [swap] {\small $\frac{1}{4}T$} (vt1)
		  (vmid) edge node {\small $\frac{1}{4}T$} (vb6)
		  (vt1) edge node [near end] {\small $(3m+1)l$} (v0i.north)
		  (vt1) edge node [swap, near end, fill=white, inner sep=1pt] {\small $(3m+1)l$} (v7i.north)
		  (v0i) edge node {$$} (v0)
		  (v0) edge node {$$} (v0o)
		  (v7i) edge node {$$} (v7)
		  (v7) edge node {$$} (v7o)
		  (v7o.south) edge node [swap, near start, fill=white, inner sep=1pt] {\small $(3m+1)l$} (vb1)
		  (v0o.south) edge node [near start] {\small $(3m+1)l$} (vb1)

		  (vt2) edge node [near end] {\small $(3m+1)l$} (v0i2.north)
		  (vt2) edge node [swap, near end] {\small $(3m+1)l$} (v7i2.north)
		  (v0i2) edge node {$$} (v02)
		  (v02) edge node {$$} (v0o2)
		  (v7i2) edge node {$$} (v72)
		  (v72) edge node {$$} (v7o2)
		  (v7o2.south) edge node [swap, near start] {\small $(3m+1)l$} (vb2)
		  (v0o2.south) edge node [near start] {\small $(3m+1)l$} (vb2)

		  (vt3) edge node [near end] {\small $(3m+1)l$} (v0i3.north)
		  (vt3) edge node [swap, near end] {\small $(3m+1)l$} (v7i3.north)
		  (v0i3) edge node {$$} (v03)
		  (v03) edge node {$$} (v0o3)
		  (v7i3) edge node {$$} (v73)
		  (v73) edge node {$$} (v7o3)
		  (v7o3.south) edge node [swap, near start] {\small $(3m+2)l$} (vb3)
		  (v0o3.south) edge node [near start] {\small $(3m+2)l$} (vb3)

		  (vt4) edge node [near end] {\small $(3m+2)l$} (v0i4.north)
		  (vt4) edge node [swap, near end] {\small $(3m+2)l$} (v7i4.north)
		  (v0i4) edge node {$$} (v04)
		  (v04) edge node {$$} (v0o4)
		  (v7i4) edge node {$$} (v74)
		  (v74) edge node {$$} (v7o4)
		  (v7o4.south) edge node [swap, near start] {\small $(3m+2)l$} (vb4)
		  (v0o4.south) edge node [near start] {\small $(3m+2)l$} (vb4)

		  (vt5) edge node [near end] {\small $(3m+2)l$} (v0i5.north)
		  (vt5) edge node [swap, near end] {\small $(3m+2)l$} (v7i5.north)
		  (v0i5) edge node {$$} (v05)
		  (v05) edge node {$$} (v0o5)
		  (v7i5) edge node {$$} (v75)
		  (v75) edge node {$$} (v7o5)
		  (v7o5.south) edge node [swap, near start] {\small $(3m+2)l$} (vb5)
		  (v0o5.south) edge node [near start] {\small $(3m+2)l$} (vb5)

		  (vt6) edge node [near end] {\small $(3m+2)l$} (v0i6.north)
		  (vt6) edge node [swap, near end, fill=white, inner sep=1pt] {\small $(3m+2)l$} (v7i6.north)
		  (v0i6) edge node {$$} (v06)
		  (v06) edge node {$$} (v0o6)
		  (v7i6) edge node {$$} (v76)
		  (v76) edge node {$ $} (v7o6)
		  (v7o6.south) edge node [swap, near start, fill=white, inner sep=1pt] {\small $(3m+2)l$} (vb6)
		  (v0o6.south) edge node [near start] {\small $(3m+2)l$} (vb6)

;

\end{tikzpicture}
\caption{An example with $m = 3$. Solid circles denote shared vertices $S = \{v_1, \ldots, v_5\}$.}
\label{fig:stack}
\end{figure}

The rest of this section is devoted to the proof of the following proposition.
\begin{proposition}\label{prop:iff}
$\bigwedge_{j \geq 0} \varphi(j)$ is satisfiable iff $G$ has a solution.
\end{proposition}
\subsection{The Proof of Proposition~\ref{prop:iff}}

We first prove the forward direction. Given a satisfying assignment of
$\bigwedge_{j \geq 0} \varphi(j)$, we construct a solution $s$ as follows: $s$ starts from $v_{top}$
and goes through the variable gadgets for $x_1^0, x_2^0, \ldots, x_m^0, x_1^1, x_2^1, \ldots, x_m^1$ in order,
eventually reaching $v_{bot}$. Each variable gadget is traversed according to the
truth value assigned to its corresponding variable. In such a traversal,
both $\mathit{pvt}_i^L$ and $\mathit{pvt}_i^R$ are visited once
(see the thick arrows in Figure~\ref{fig:connectaux} for the situation when $x_i^0$ is assigned $\mathbf{true}$
and $x_i^1$ is assigned $\mathbf{false}$).
Along the way from $v_{top}$ to $v_{bot}$, $s$ detours at certain times and `hits' each clause vertex
exactly once as illustrated by the thick arrows in Figure~\ref{fig:detour}
(this can be done as $\varphi(0)$ is satisfied by the assignment).
Then $s$ goes back to $v_{top}$ through $v_{mid}$ and starts over again, this time following 
the truth values assigned to variables in $\overline{x}^1 \cup \overline{x}^2$, and so on.
One can verify that this describes a solution to $G$.

Now consider the other direction. Let 
\[
s = (v_{mid} s_1 v_{mid} \ldots v_{mid} s_p)^\omega
\]
be a periodic solution to $G$ where each \emph{segment} $s_j$, $j \in \{1, \ldots, p \}$ is a
finite subpath visiting only vertices in $V \setminus \{ v_{mid} \}$.

%

\begin{proposition}\label{prop:tborbt}
In $s = (v_{mid} s_1 v_{mid} \ldots v_{mid} s_p)^\omega$, either of the following holds:
\begin{itemize}
\item All $s_j$, $j \in \{ 1, \ldots, p \}$ starts with $v_{top}$ and ends with $v_{bot}$
\item All $s_j$, $j \in \{ 1, \ldots, p \}$ starts with $v_{bot}$ and ends with $v_{top}$.
\end{itemize}
\end{proposition}
\begin{proof}
See Appendix~\ref{app:tborbt}.
\end{proof}
We therefore further assume that $s$ satisfies the first case of the proposition above
(this is sound as a periodic solution can be `reversed' while remaining a valid solution).
We argue that $s$ `witnesses' a satisfying assignment of $\bigwedge_{j \geq 0} \varphi(j)$.
\todo{should this be called a lemma?}
\begin{proposition}\label{prop:exact}
In each segment $s_j$, each vertex in \( \bigcup_{i \in \{1, \ldots, m\}} \{\mathit{pvt}_i^L, \mathit{pvt}_i^R\} \) appears twice
whereas other vertices in $V \setminus \{ v_{mid} \}$ appear once.
\end{proposition}
\begin{proof}
See Appendix~\ref{app:cnt}.
\end{proof}
Based on this proposition, we show that $s$ cannot `jump' between variable gadgets via clause vertices.
It follows that the traversal of each $\mathit{Row}_i$ must be done in a single pass.

\begin{proposition}\label{prop:nojump}
In each segment $s_j$, if $v^{c_k}$ is entered from a clause box (in some variable gadget), the edge that
immediately follows must go back to the same clause box.
\end{proposition}
\begin{proof}
Consider a $3 \times 3$ `box' formed by a separator box
and (the left- or right-) half of a clause box.
Note that except for the four vertices at the corners, no vertex in this $3 \times 3$ box is 
connected to the rest of the graph. Recall that if each vertex in this $3 \times 3$ box is to be
visited only once (as enforced by Proposition~\ref{prop:exact}),
it must be traversed in the patterns illustrated in Figures~\ref{fig:pat1} and~\ref{fig:pat2}.

\begin{figure}[h]
\begin{minipage}[b]{0.5\linewidth}
\centering
\begin{tikzpicture}[-,>=stealth', auto, transform shape, node distance=2.5cm,
                    semithick, every state/.style={fill=none,draw=black,text=black,shape=circle,scale=0.3}]

\node[state, draw=none]	   (vi) 								{\LARGE $\boldsymbol{\cdots}$};
\node[state, draw=none]	   (vii) [above=.37cm of vi]	{\LARGE $\boldsymbol{\cdots}$};
\node[state, draw=none]	   (vio) [below=.37cm of vi]	{\LARGE $\boldsymbol{\cdots}$};

\node[state]			   (v7) [left of =vi] {$$};
\node[state]			   (v7i) [below=.5cm of v7] {$$};
\node[state]			   (v7o) [above=.5cm of v7] {$$};

\node[state]			   (v6) [left of =v7] {$$};
\node[state]			   (v6i) [above=.5cm of v6] {$$};
\node[state]			   (v6o) [below=.5cm of v6] {$$};

\node[state]			   (v5) [left of =v6] {$$};
\node[state]			   (v5i) [above=.5cm of v5] {$$};
\node[state]			   (v5o) [below=.5cm of v5] {$$};

\node[state, draw=none]	   (vi') [left of =v5]				{\LARGE $\boldsymbol{\cdots}$};
\node[state, draw=none]	   (vii') [above=.37cm of vi']	{\LARGE $\boldsymbol{\cdots}$};
\node[state, draw=none]	   (vio') [below=.37cm of vi']	{\LARGE $\boldsymbol{\cdots}$};


\path   (vii') edge node {\scriptsize $$} (v5i)
        (vio') edge node [swap] {\scriptsize $$} (v5o)
		  (v7o) edge node {\scriptsize $$} (vii)
        (vio) edge node {\scriptsize $$} (v7i)
			
        (v5) edge [line width=1mm] node {\scriptsize $$} (v5i)
        (v5o) edge [line width=1mm] node {\scriptsize $$} (v5)

        (v6) edge [line width=1mm] node {\scriptsize $$} (v6i)
        (v6o) edge [line width=1mm] node {\scriptsize $$} (v6)

        (v7) edge [line width=1mm] node {\scriptsize $$} (v7i)
        (v7o) edge [line width=1mm] node {\scriptsize $$} (v7)

        (v5i) edge node {\scriptsize $$} (v6i)
        (v6o) edge [line width=1mm] node {\scriptsize $$} (v5o)

        (v6i) edge [line width=1mm] node {\scriptsize $$} (v7o)
        (v7i) edge node {\scriptsize $$} (v6o);

\end{tikzpicture}
\caption{Pattern `$\sqcupsqcap$'}
\label{fig:pat1}
\end{minipage}
\begin{minipage}[b]{0.5\linewidth}
\centering
\begin{tikzpicture}[-,>=stealth', auto, transform shape, node distance=2.5cm,
                    semithick, every state/.style={fill=none,draw=black,text=black,shape=circle,scale=0.3}]

\node[state, draw=none]	   (vi) 								{\LARGE $\boldsymbol{\cdots}$};
\node[state, draw=none]	   (vii) [above=.37cm of vi]	{\LARGE $\boldsymbol{\cdots}$};
\node[state, draw=none]	   (vio) [below=.37cm of vi]	{\LARGE $\boldsymbol{\cdots}$};

\node[state]			   (v7) [left of =vi] {$$};
\node[state]			   (v7i) [below=.5cm of v7] {$$};
\node[state]			   (v7o) [above=.5cm of v7] {$$};

\node[state]			   (v6) [left of =v7] {$$};
\node[state]			   (v6i) [above=.5cm of v6] {$$};
\node[state]			   (v6o) [below=.5cm of v6] {$$};

\node[state]			   (v5) [left of =v6] {$$};
\node[state]			   (v5i) [above=.5cm of v5] {$$};
\node[state]			   (v5o) [below=.5cm of v5] {$$};

\node[state, draw=none]	   (vi') [left of =v5]				{\LARGE $\boldsymbol{\cdots}$};
\node[state, draw=none]	   (vii') [above=.37cm of vi']	{\LARGE $\boldsymbol{\cdots}$};
\node[state, draw=none]	   (vio') [below=.37cm of vi']	{\LARGE $\boldsymbol{\cdots}$};


\path   (vii') edge node {\scriptsize $$} (v5i)
        (vio') edge node [swap] {\scriptsize $$} (v5o)
		  (v7o) edge node {\scriptsize $$} (vii)
        (vio) edge node {\scriptsize $$} (v7i)
			
        (v5) edge [line width=1mm] node {\scriptsize $$} (v5i)
        (v5o) edge [line width=1mm] node {\scriptsize $$} (v5)

        (v6) edge [line width=1mm] node {\scriptsize $$} (v6i)
        (v6o) edge [line width=1mm] node {\scriptsize $$} (v6)

        (v7) edge [line width=1mm] node {\scriptsize $$} (v7i)
        (v7o) edge [line width=1mm] node {\scriptsize $$} (v7)

        (v5i) edge [line width=1mm] node {\scriptsize $$} (v6i)
        (v6o) edge node {\scriptsize $$} (v5o)

        (v6i) edge node {\scriptsize $$} (v7o)
        (v7i) edge [line width=1mm] node {\scriptsize $$} (v6o);

\end{tikzpicture}
\caption{Pattern `$\sqcapsqcup$'}
\label{fig:pat2}
\end{minipage}
\end{figure}

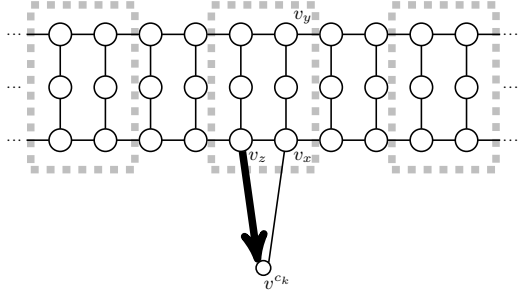
\begin{figure}[h]
\centering
\begin{tikzpicture}[-,>=stealth', auto, transform shape, node distance=2.5cm, scale=0.8,
                    semithick, every state/.style={fill=none,draw=black,text=black,shape=circle,scale=0.3}]

\node[state]	   (vi) 								{\phantom{\LARGE $\boldsymbol{\cdots}$}};
\node[state]	   (vii) [above=.5cm of vi]	{\phantom{\LARGE $\boldsymbol{\cdots}$}};
\node[state]	   (vio) [below=.5cm of vi]	{\phantom{\LARGE $\boldsymbol{\cdots}$}};

\node[state]	   (viinv) 	[right of =vi]				{\phantom{\LARGE $\boldsymbol{\cdots}$}};
\node[state]	   (viiinv) [above=.5cm of viinv]	{\phantom{\LARGE $\boldsymbol{\cdots}$}};
\node[state]	   (vioinv) [below=.5cm of viinv]	{\phantom{\LARGE $\boldsymbol{\cdots}$}};

\node[state, draw=none]	   (viinv2) 	[right of =viinv]			{\LARGE $\boldsymbol{\cdots}$};
\node[state, draw=none]	   (viiinv2) [above=.5cm of viinv2]	{\LARGE $\boldsymbol{\cdots}$};
\node[state, draw=none]	   (vioinv2) [below=.5cm of viinv2]	{\LARGE $\boldsymbol{\cdots}$};

\node[state]			   (v7) [left of =vi] 		{\phantom{\LARGE $\boldsymbol{\cdots}$}}; 
\node[state]			   (v7i) [below=.5cm of v7] {\phantom{\LARGE $\boldsymbol{\cdots}$}}; 
\node[state]			   (v7o) [above=.5cm of v7] {\phantom{\LARGE $\boldsymbol{\cdots}$}}; 

\node[state]			   (v6) [left of =v7] 			{\phantom{\LARGE $\boldsymbol{\cdots}$}}; 
\node[state]			   (v6i) [above=.5cm of v6] {\phantom{\LARGE $\boldsymbol{\cdots}$}}; 
\node[state]			   (v6o) [below=.5cm of v6] {\phantom{\LARGE $\boldsymbol{\cdots}$}}; 

\node[state]			   (v5) [left of =v6] 		  {\phantom{\LARGE $\boldsymbol{\cdots}$}};  
\node[state]			   (v5i) [above=.5cm of v5] {\phantom{\LARGE $\boldsymbol{\cdots}$}};    
\node[state]			   (v5o) [below=.5cm of v5] {\phantom{\LARGE $\boldsymbol{\cdots}$}};    

\node[state]			   (v4) [left of =v5] 		 {\phantom{\LARGE $\boldsymbol{\cdots}$}};    
\node[state]			   (v4i) [below=.5cm of v4] {\phantom{\LARGE $\boldsymbol{\cdots}$}};     
\node[state]			   (v4o) [above=.5cm of v4] {\phantom{\LARGE $\boldsymbol{\cdots}$}};     

\node[state]			   (v3) [left of =v4] 			{\phantom{\LARGE $\boldsymbol{\cdots}$}};      
\node[state]			   (v3i) [above=.5cm of v3] {\phantom{\LARGE $\boldsymbol{\cdots}$}};      
\node[state]			   (v3o) [below=.5cm of v3] {\phantom{\LARGE $\boldsymbol{\cdots}$}};      

\node[state]			   (v2) [left of =v3]		  {\phantom{\LARGE $\boldsymbol{\cdots}$}};       
\node[state]			   (v2i) [below=.5cm of v2] {\phantom{\LARGE $\boldsymbol{\cdots}$}};         
\node[state]			   (v2o) [above=.5cm of v2] {\phantom{\LARGE $\boldsymbol{\cdots}$}};         

\node[state]	   (vi') [left of =v2]			 {\phantom{\LARGE $\boldsymbol{\cdots}$}};       
\node[state]	   (vii') [above=.5cm of vi']	{\phantom{\LARGE $\boldsymbol{\cdots}$}};         
\node[state]	   (vio') [below=.5cm of vi']	{\phantom{\LARGE $\boldsymbol{\cdots}$}};         

\node[state]	   (vi'inv) 	[left of =vi']			 {\phantom{\LARGE $\boldsymbol{\cdots}$}};       
\node[state]	   (vii'inv) [above=.5cm of vi'inv]	{\phantom{\LARGE $\boldsymbol{\cdots}$}};         
\node[state]	   (vio'inv) [below=.5cm of vi'inv]	{\phantom{\LARGE $\boldsymbol{\cdots}$}};         

\node[state, draw=none]	   (vi'inv2) 	[left of =vi'inv]			 {\LARGE $\boldsymbol{\cdots}$};       
\node[state, draw=none]	   (vii'inv2) [above=.5cm of vi'inv2]	{\LARGE $\boldsymbol{\cdots}$};         
\node[state, draw=none]	   (vio'inv2) [below=.5cm of vi'inv2]	{\LARGE $\boldsymbol{\cdots}$};         

\node[state]				 (cj) [below=2cm of $(v4i)!0.5!(v5o)$]  {$$};

\node (box1) [draw=gray!50, fit=(v4o) (v5o), dashed, line width=1mm, inner sep=3mm] {};
\node (box2) [draw=gray!50, fit=(vio) (viiinv), dashed, line width=1mm, inner sep=3mm] {};
\node (box2) [draw=gray!50, fit=(vio') (vii'inv), dashed, line width=1mm, inner sep=3mm] {};

\node [xshift=1.5mm, yshift=-1.5mm] at (cj.south east) {$v^{c_k}$};
\node [xshift=1.5mm, yshift=1.5mm] at (v5i.north east) {$v_y$};
\node [xshift=1.5mm, yshift=-1.5mm] at (v5o.south east) {$v_x$};
\node [xshift=1.5mm, yshift=-1.5mm] at (v4i.south east) {$v_z$};

\path   (vii') edge node {\scriptsize $$} (v2o)
        (vio') edge node [swap] {\scriptsize $$} (v2i)
		  (v7o) edge node {\scriptsize $$} (vii)
        (vio) edge node {\scriptsize $$} (v7i)

        (vi) edge node {\scriptsize $$} (vii)
        (vi) edge node {\scriptsize $$} (vio)

        (viinv) edge node {\scriptsize $$} (viiinv)
        (viinv) edge node {\scriptsize $$} (vioinv)

        (vi') edge node {\scriptsize $$} (vii')
        (vi') edge node {\scriptsize $$} (vio')

        (vi'inv) edge node {\scriptsize $$} (vii'inv)
        (vi'inv) edge node {\scriptsize $$} (vio'inv)

        (vii) edge node {\scriptsize $$} (viiinv)
        (vio) edge node {\scriptsize $$} (vioinv)

        (vii') edge node {\scriptsize $$} (vii'inv)
        (vio') edge node {\scriptsize $$} (vio'inv)

        (viiinv) edge node {\scriptsize $$} (viiinv2)
        (vioinv) edge node {\scriptsize $$} (vioinv2)

        (vii'inv) edge node {\scriptsize $$} (vii'inv2)
        (vio'inv) edge node {\scriptsize $$} (vio'inv2)

        (v2) edge node {\scriptsize $$} (v2i)
        (v2o) edge node {\scriptsize $$} (v2)

        (v3) edge node {\scriptsize $$} (v3i)
        (v3o) edge node {\scriptsize $$} (v3)

        (v4) edge node {\scriptsize $$} (v4i)
        (v4o) edge node {\scriptsize $$} (v4)

        (v5) edge node {\scriptsize $$} (v5i)
        (v5o) edge node {\scriptsize $$} (v5)

        (v6) edge node {\scriptsize $$} (v6i)
        (v6o) edge node {\scriptsize $$} (v6)

        (v7) edge node {\scriptsize $$} (v7i)
        (v7o) edge node {\scriptsize $$} (v7)

        (v2i) edge node {\scriptsize $$} (v3o)
        (v3i) edge node {\scriptsize $$} (v2o)

        (v3i) edge node {\scriptsize $$} (v4o)
        (v4i) edge node {\scriptsize $$} (v3o)

        (v4i) edge node {\scriptsize $$} (v5o)
        (v5i) edge node {\scriptsize $$} (v4o)

        (v5i) edge node {\scriptsize $$} (v6i)
        (v6o) edge node {\scriptsize $$} (v5o)

        (v6i) edge node {\scriptsize $$} (v7o)
        (v7i) edge node {\scriptsize $$} (v6o)

        (v4i) edge [->, line width=1mm] node [swap] {\scriptsize $$} (cj.north west)
        (cj.north east) edge node [swap] {\scriptsize $$} (v5o);

\end{tikzpicture}
\caption{$x_i^0$ occurs positively in $c_k$}
\label{fig:cannotjump}
\end{figure}

Now consider the situation in Figure~\ref{fig:cannotjump} where $s_j$ goes from
$v_z$ to $v^{c_k}$. The $3 \times 3$ box with $v_z$ at its lower-right
must be traversed in Pattern `$\sqcupsqcap$' (as otherwise $v_z$ will be visited twice).
Assume that $s_j$ does not visit $v_x$ immediately after $v^{c_k}$.
As $v_x$ cannot be entered or left via $v_z$ and $v^{c_k}$, the $3 \times 3$ box with $v_x$ at its lower-left
must also be traversed in Pattern `$\sqcupsqcap$'. However, there is then no way to enter or leave $v_y$.
This is a contradiction.
\end{proof}
Note that in Figure~\ref{fig:cannotjump}, the three clause boxes (framed by dotted lines) are all traversed
in Pattern `$\sqcap$' or they are all traversed in Pattern `$\sqcup$'. More generally, we
have the following proposition. 
\begin{proposition}\label{col:asawhole}
In each segment $s_j$, clause boxes in a given variable gadget
are all traversed in Pattern `$\sqcap$' or they are all traversed in Pattern `$\sqcup$' (with possible detours via clause vertices).
\end{proposition}

Write $v \rightarrow v'$ for the edge from $v$ to $v'$ and
$v \leadsto v'$ for a finite path that starts with $v$ and ends with $v'$.
By Proposition~\ref{prop:exact}, each segment $s_j$ can be written as $v_{top} \leadsto v_{b_1} \leadsto \cdots \leadsto v_{b_{2m - 1}} \leadsto v_{bot}$
where $b_1, \ldots, b_{2m - 1}$ is a permutation of $1, \ldots, 2m - 1$.
We show that each subpath $v \leadsto v'$ of $s_j$ with distinct $v, v' \in S \cup \{v_{top}, v_{bot}\}$
and no $v'' \in S \cup \{ v_{top}, v_{bot} \}$ in between
must be of a very restricted form. For convenience, we call such a subpath $v \leadsto v'$ a \emph{fragment}.

\begin{proposition}\label{prop:twopivots}
In each segment $s_j = v_{top} \leadsto v_{b_1} \leadsto \cdots \leadsto v_{b_{2m - 1}} \leadsto v_{bot}$,
a fragment $v \leadsto v'$ visits $\mathit{pvt}_i^L$ and $\mathit{pvt}_i^R$ (once for each) for some $i \in \{1, \ldots, m\}$.
Moreover, each fragment $v \leadsto v'$ in $v_{top} \leadsto v_{b_1} \leadsto \cdots \leadsto v_{b_m}$
visits a different set $\{\mathit{pvt}_i^L, \mathit{pvt}_i^R\}$.
The same holds for $v_{b_m} \leadsto v_{b_{m+1}} \leadsto \cdots \leadsto v_{bot}$.
\end{proposition}
\begin{proof}
It is clear that $\mathit{dur}(v \leadsto v') \geq 2(3m+1)l$,
and hence $\mathit{dur}(v_{top} \leadsto v_{b_1} \leadsto \cdots v_{b_m}) \geq m \big(2(3m+1)l\big)$.
Let there be a vertex \(v \in \bigcup_{i \in \{1, \ldots, m\}} \{\mathit{pvt}_i^L, \mathit{pvt}_i^R\}\)
missing in $v_{top} \leadsto v_{b_1} \leadsto \cdots v_{b_m}$.
Since the time needed from $v_{b_m}$ to $v$ is greater than $(3m+1)l$,
even if $s_j$ visits $v$ as soon as possible after $v_{b_m}$, 
the duration from $v_{bot}$ in $s_{j - 1}$ to $v$ in $s_j$ will still be greater than
$\frac{1}{2}T + m\big(2(3m+1)l\big) + (3m + 1)l > \mathit{RD}(v)$, which is a contradiction.
Therefore, all vertices in \( \bigcup_{i \in \{1, \ldots, m\}} \{\mathit{pvt}_i^L, \mathit{pvt}_i^R\} \) must appear in the subpath from $v_{top}$ to $v_{b_m}$.
The same holds for the subpath from $v_{b_m}$ to $v_{bot}$ by similar arguments.
Now note that by Proposition~\ref{prop:nojump}, a fragment $v \leadsto v'$ may visit
at most two vertices---$\{\mathit{pvt}_i^L, \mathit{pvt}_i^R\}$ for some $i \in \{1, \ldots, m\}$.
The proposition then follows from Proposition~\ref{prop:exact}.
\end{proof}

\begin{proposition}\label{prop:onerow}
In each segment $s_j$, a fragment $v \leadsto v'$
visits all vertices in either $\mathit{Row}_i$ or $\mathit{Row}_{i + m}$ for some $i \in \{1, \ldots, m\}$
but not a single vertex in $\bigcup_{\substack{j \neq i \\ j \in \{1, \ldots, m\}}} (\mathit{Row}_j \cup \mathit{Row}_{j + m})$.
\end{proposition}


Now consider a fragment $v \leadsto v'$ that visits $\mathit{pvt}_i^L$ and $\mathit{pvt}_i^R$ (by Proposition~\ref{prop:twopivots}).
By Proposition~\ref{prop:exact}, $v \leadsto v'$ must also visit exactly two vertices other than $\mathit{pvt}_i^L$ in $\mathit{LCG}_i$
and exactly two vertices other than $\mathit{pvt}_i^R$ in $\mathit{RCG}_i$ (once for each).
It is not hard to see that $v \leadsto v'$ must contain, in order, the following subpaths (together with
some obvious choices of edges connecting these subpaths):
\begin{enumerate}[(i).]
\item \label{itm:one} A long edge, e.g., $v_i \rightarrow v_i^{b, R}$.
\item A `side', e.g., $v_i^{b, R} \rightarrow v_i^{m, R} \rightarrow v_i^{t, R}$.
\item A subpath consisting of a $\mathit{pvt}$ vertex and two other vertices in the relevant consistency gadget,
e.g., $\mathit{out}^{\uparrow, R}_i \rightarrow \mathit{pvt}^R_i \rightarrow \mathit{in}^{\downarrow, R}_i$.
\item \label{itm:four} A traversal of a row with detours.
\item A subpath consisting of a $\mathit{pvt}$ vertex and two other vertices in the relevant consistency gadget.
\item A side.
\item \label{itm:seven} A long edge.
\end{enumerate}
The following proposition is then immediate. In particular, the exact value of $\mathit{dur}(v \leadsto v')$ is \mbox{decided by}:
\begin{itemize}
\item $\mathit{FT}$ of the long edges taken in (\ref{itm:one}) and (\ref{itm:seven})
\item detours to clause vertices in (\ref{itm:four}).
\end{itemize}
\begin{proposition}\label{prop:dur}
In each segment $s_j$, the following holds for all fragments \mbox{$v \leadsto v'$}:
\[
2 (3m+1)l + l \leq \mathit{dur}(v \leadsto v') \leq 2 (3m+2)l + l + 2h.
\]
\end{proposition}

\begin{proposition}\label{prop:sameorder}
The order the sets $\{\mathit{pvt}_i^L, \mathit{pvt}_i^R\}$ are visited (regardless
of which vertex in the set is first visited) in the first $m$ fragments of each segment $s_j$
is identical to the order they are visited in the last $m$ fragments of $s_{j-1}$.
\end{proposition}
\begin{proof}
By Proposition~\ref{prop:dur}, if this does not hold then there must be a $\mathit{pvt}$ vertex
having two occurrences in $s$ separated by more than $\frac{1}{2}T + m\big(2 (3m+1)l + l\big) + 2(3m+1)l$.
This is a contradiction.
\end{proof}

For each segment $s_j$, we denote by $\mathit{first}(s_j)$ the `first half' of $s_j$, i.e., the subpath
of $s_j$ that consists of the first $m$ fragments of $s_j$ and by $\mathit{second}(s_j)$ the `second half' of $s_j$.
Write $\exists (v \leadsto v') \subseteq u$ if $u$ has a subpath of the form $v \leadsto v'$.

\begin{proposition}\label{prop:inorder}
In each segment $s_j = v_{top} \leadsto v_{b_1} \leadsto \cdots \leadsto v_{b_{2m - 1}} \leadsto v_{bot}$,
we have $b_i = i$ for all $i \in \{1, \ldots, 2m - 1\}$.
\end{proposition}
\begin{proof}
First note that by construction and Proposition~\ref{prop:twopivots}, $\{\mathit{pvt}_m^L, \mathit{pvt}_m^R\}$ must be the last set of $\mathit{pvt}$ vertices visited in $\mathit{second}(s_{j - 1})$.
By Proposition~\ref{prop:sameorder}, it must also be the last set of $\mathit{pvt}$ vertices visited in $\mathit{first}(s_j)$.
Now assume that a long edge of flight time $(3m + 2)l$ is taken before $\mathit{pvt}_m^L$ and $\mathit{pvt}_m^R$ are visited
in $\mathit{first}(s_j)$.
Consider the following cases:
\begin{itemize}
\item $\exists (\mathit{pvt}_m^L \leadsto \mathit{pvt}_m^R) \subseteq \mathit{second}(s_{j-1})$ and $\exists (\mathit{pvt}_m^R \leadsto \mathit{pvt}_m^L) \subseteq \mathit{first}(s_{j})$:
Note that the last edge taken in $s_{j-1}$ is a long edge of flight time $(3m + 2)l$,
and hence there are two occurrences of $\mathit{pvt}_m^L$ in $s$ separated by at least
$\frac{1}{2}T + m\big(2 (3m+1)l + l\big) + 2l > \frac{1}{2}T + m\big(2 (3m+1)l + l\big) + l + 4h = \mathit{RD}(\mathit{pvt}_m^L)$.
\item $\exists (\mathit{pvt}_m^R \leadsto \mathit{pvt}_m^L) \subseteq \mathit{second}(s_{j-1})$ and $\exists (\mathit{pvt}_m^L \leadsto \mathit{pvt}_m^R) \subseteq \mathit{first}(s_{j})$:
The same argument shows that $\mathit{pvt}_m^R$ must miss its relative deadline.
\item $\exists (\mathit{pvt}_m^L \leadsto \mathit{pvt}_m^R) \subseteq \mathit{second}(s_{j-1})$ and $\exists (\mathit{pvt}_m^L \leadsto \mathit{pvt}_m^R) \subseteq \mathit{first}(s_{j})$:
The same argument shows that both $\mathit{pvt}_m^L$ and $\mathit{pvt}_m^R$ must miss their relative deadlines.
\item $\exists (\mathit{pvt}_m^R \leadsto \mathit{pvt}_m^L) \subseteq \mathit{second}(s_{j-1})$ and $\exists (\mathit{pvt}_m^R \leadsto \mathit{pvt}_m^L) \subseteq \mathit{first}(s_{j})$:
The same argument shows that both $\mathit{pvt}_m^L$ and $\mathit{pvt}_m^R$ must miss their relative deadlines.
\end{itemize}
We therefore conclude that in $\mathit{first}(s_j)$, all long edges taken before $\mathit{pvt}_m^L$ and $\mathit{pvt}_m^R$ are visited
must have $FT$ equal to $(3m+1)l$.
Furthermore, all such long edges must be traversed `downwards' (by Proposition~\ref{prop:exact}).
It follows that $b_i = i$ for $i \in \{1, \ldots, m - 1\}$.
By Proposition~\ref{prop:sameorder}, Proposition~\ref{prop:exact} and $m > 2$, we easily derive that $b_m = m$
and then $b_i = i$ for $i \in \{m + 1, \ldots, 2m - 1\}$.
\end{proof}

By Proposition~\ref{prop:inorder}, the long edges in each variable
gadget must be traversed in the ways shown in Figures~\ref{fig:positive} and~\ref{fig:negative}.
\begin{figure}[h]
\begin{minipage}[b]{0.5\linewidth}
\centering
\begin{tikzpicture}[-,>=stealth', auto, transform shape, node distance=.6cm,
                    semithick, every state/.style={fill=none,draw=black,text=black,shape=circle,scale=0.15}]

\node[state, draw=none]			   (v3) {$$};
\node[state, draw=none]			   (v3i) [above=.25cm of v3] {$$};
\node[state, draw=none]			   (v3o) [below=.25cm of v3] {$$};

\node[state, draw=none]			   (v2) [left of =v3] {$$};
\node[state, draw=none]			   (v2i) [below=.25cm of v2] {$$};
\node[state, draw=none]			   (v2o) [above=.25cm of v2] {$$};

\node[state, draw=none]			   (v1) [left of =v2] {$$};
\node[state, draw=none]			   (v1i) [above=.25cm of v1] {$$};
\node[state, draw=none]			   (v1o) [below=.25cm of v1] {$$};

\node[state]			   (v0) [left=1cm of v1] {$$};
\node[state]			   (v0i) [above=.4cm of v0] {$$};
\node[state]			   (v0o) [below=.4cm of v0] {$$};

\node[state]				(vt1)	[above right=1cm and 1.25cm of v2] {$$};
\node[state]				(vb1)	[below right=1cm and 1.25cm of v2] {$$};

\node[state, draw=none]			   (v5) [below right=1cm and 1.25cm of vt1] {$$};
\node[state, draw=none]			   (v5i) [below=.25cm of v5] {$$};
\node[state, draw=none]			   (v5o) [above=.25cm of v5] {$$};

\node[state, draw=none]			   (v4) [left of =v5] {$$};
\node[state, draw=none]			   (v4i) [above=.25cm of v4] {$$};
\node[state, draw=none]			   (v4o) [below=.25cm of v4] {$$};

\node[state, draw=none]			   (v6) [right of =v5] {$$};
\node[state, draw=none]			   (v6i) [above=.25cm of v6] {$$};
\node[state, draw=none]			   (v6o) [below=.25cm of v6] {$$};

\node[state]			   (v7) [right=1cm of v6] {$$};
\node[state]			   (v7i) [above=.4cm of v7] {$$};
\node[state]			   (v7o) [below=.4cm of v7] {$$};

\path   (vt1) edge [->, line width=1mm] node {$$} (v0i.north)
		  (v0i) edge [->, line width=1mm] node {$$} (v0)
		  (v0) edge [->, line width=1mm] node {$$} (v0o)
		  (v7i) edge [->, line width=1mm] node {$$} (v7)
		  (v7) edge [->, line width=1mm] node {$$} (v7o)
		  (v7o.south) edge [->, line width=1mm] node {$$} (vb1);

\end{tikzpicture}
\caption{The variable is assigned to $\mathbf{true}$}
\label{fig:positive}
\end{minipage}
\begin{minipage}[b]{0.5\linewidth}
\centering
\begin{tikzpicture}[-,>=stealth', auto, transform shape, node distance=.6cm,
                    semithick, every state/.style={fill=none,draw=black,text=black,shape=circle,scale=0.15}]

\node[state, draw=none]			   (v3) {$$};
\node[state, draw=none]			   (v3i) [above=.25cm of v3] {$$};
\node[state, draw=none]			   (v3o) [below=.25cm of v3] {$$};

\node[state, draw=none]			   (v2) [left of =v3] {$$};
\node[state, draw=none]			   (v2i) [below=.25cm of v2] {$$};
\node[state, draw=none]			   (v2o) [above=.25cm of v2] {$$};

\node[state, draw=none]			   (v1) [left of =v2] {$$};
\node[state, draw=none]			   (v1i) [above=.25cm of v1] {$$};
\node[state, draw=none]			   (v1o) [below=.25cm of v1] {$$};

\node[state]			   (v0) [left=1cm of v1] {$$};
\node[state]			   (v0i) [above=.4cm of v0] {$$};
\node[state]			   (v0o) [below=.4cm of v0] {$$};

\node[state]				(vt1)	[above right=1cm and 1.25cm of v2] {$$};
\node[state]				(vb1)	[below right=1cm and 1.25cm of v2] {$$};

\node[state, draw=none]			   (v5) [below right=1cm and 1.25cm of vt1] {$$};
\node[state, draw=none]			   (v5i) [below=.25cm of v5] {$$};
\node[state, draw=none]			   (v5o) [above=.25cm of v5] {$$};

\node[state, draw=none]			   (v4) [left of =v5] {$$};
\node[state, draw=none]			   (v4i) [above=.25cm of v4] {$$};
\node[state, draw=none]			   (v4o) [below=.25cm of v4] {$$};

\node[state, draw=none]			   (v6) [right of =v5] {$$};
\node[state, draw=none]			   (v6i) [above=.25cm of v6] {$$};
\node[state, draw=none]			   (v6o) [below=.25cm of v6] {$$};

\node[state]			   (v7) [right=1cm of v6] {$$};
\node[state]			   (v7i) [above=.4cm of v7] {$$};
\node[state]			   (v7o) [below=.4cm of v7] {$$};

\path   (vt1) edge [->, line width=1mm] node {$$} (v7i.north)
		  (v0i) edge [->, line width=1mm] node {$$} (v0)
		  (v0) edge [->, line width=1mm] node {$$} (v0o)
		  (v7i) edge [->, line width=1mm] node {$$} (v7)
		  (v7) edge [->, line width=1mm] node {$$} (v7o)
		  (v0o.south) edge [->, line width=1mm] node {$$} (vb1);

\end{tikzpicture}
\caption{The variable is assigned to $\mathbf{false}$}
\label{fig:negative}
\end{minipage}
\end{figure}

\vspace{-0.5cm}

\begin{proposition}\label{prop:consistent}
For each segment $s_j$, the ways in which the long edges are traversed 
in the last $m$ fragments of $s_j$ are consistent with
the ways in which the long edges are traversed in the first $m$ fragments
of $s_{j+1}$.
\end{proposition}
\begin{proof}
Without loss of generality, consider the case that
$\exists (\mathit{pvt}_i^L \leadsto \mathit{pvt}_i^R) \subseteq \mathit{second}(s_{j})$
and $\exists (\mathit{pvt}_i^R \leadsto \mathit{pvt}_i^L) \subseteq \mathit{first}(s_{j+1})$.
By Proposition~\ref{prop:inorder}, these two occurrences of $\mathit{pvt}_i^L$
in $s$ are separated by, at least, the sum of $\frac{1}{2}T + m\big(2(3m+2)l + l\big)  - (2i - 1)l$
and the duration of the actual subpath $\mathit{pvt}_i^R \leadsto \mathit{pvt}_i^L$ in $\mathit{first}(s_{j+1})$.
It is clear that $\mathit{pvt}_i^L$ must miss its relative deadline.
\end{proof}

\begin{proposition}\label{prop:hitcorrect}
In each segment $s_j$, if a variable gadget is traversed as in Figure~\ref{fig:positive} (Figure~\ref{fig:negative}), 
then all of its clause boxes are traversed in Pattern `$\sqcup$' (Pattern `$\sqcap$').
\end{proposition}

Consider a segment $s_j$. As each clause vertex is visited once in $s_j$ (by Proposition~\ref{prop:exact}),
the ways in which the long edges are traversed in all fragments $v \leadsto v'$ of $s_j$ (i.e., as in Figure~\ref{fig:positive} or Figure~\ref{fig:negative})
can be seen as a satisfying assignment of $\varphi(0)$ (by construction and Proposition~\ref{prop:hitcorrect}).
By the same argument, the ways in which the long edges are traversed in all fragments of $s_{j+1}$
can be seen as a satisfying assignment of $\varphi(1)$.
Now by Proposition~\ref{prop:consistent}, the assignment of variables
$\overline{x}^1$ is consistent in both segments. 
By IH, $s$ witnesses a (periodic) satisfying assignment of
$\bigwedge_{j \geq 0} \varphi(j)$. Proposition~\ref{prop:iff} is hence proved.

Finally, note that $\mathit{FT}$ can easily be modified into a metric
over $V$ by replacing each entry of value $2T$ with the `shortest
distance' between the two relevant vertices.  It is easy to see that
Proposition~\ref{prop:iff} still holds.  Our main result, which holds
for the metric case, follows immediately from
Section~\ref{subsec:membership}.
\begin{theorem}
The \textsc{cr-uav} Problem is $\mathrm{PSPACE}$-complete.\footnote{Our result holds irrespective
of whether the numbers are encoded in unary or binary.}
\end{theorem}

\section{Conclusion}
We have proved that the \textsc{cr-uav} Problem is
$\mathrm{PSPACE}$-complete even in the single-UAV case.  The proof
reveals a 
connection between a periodically specified problem and a recurrent
path-planning problem (which is not \emph{succinctly specified} in the sense
of~\cite{Marathe1997}).  We list below some possible directions for
future work:

\begin{enumerate}

\item A number of crucial problems in other domains, e.g., the
  generalised pinwheel scheduling problem~\cite{Feinberg2005} and the message
  ferrying problem~\cite{Zhao2004}, share similarities with the
  \textsc{cr-uav} Problem---namely, they have relative deadlines and
  therefore `contexts'.  Most of these problems are only known to be
  $\mathrm{NP}$-hard. It would be interesting to investigate whether
  our construction can be adapted to establish
  $\mathrm{PSPACE}$-hardness of these problems.

\item It is claimed in \cite{Fargeas2013} that the restricted case
  in which vertices can be realised as points in a two-dimensional plane
  (with discretised distances between points) is $\mathrm{NP}$-complete (with a single UAV).
  A natural question is the relationship with the problem studied in the present paper.

\item Current approaches to solving the \textsc{cr-uav} Problem often
  formulate it as a Mixed-Integer Linear Program (MILP) and then
  invoke an off-the-shelf solver (see, e.g.,~\cite{Basilico2012}). Yet
  as implied by Proposition~\ref{prop:iff}, the length of a solution
  can however be exponential in the size of the problem
  instance. We are currently investigating alternative implementations
  which would overcome such difficulties.

\end{enumerate}

\bibliographystyle{splncs03}
\bibliography{refs}        

\newpage
\appendix

\section{A Counterexample}\label{app:cex}

In~\cite{Basilico2012} it is claimed that the \textsc{cr-uav} Problem with a
single UAV is in $\mathrm{NP}$. The claim is based on the following bound on the periods
of solutions:

\begin{claim}[{\cite[Theorem $4.5$]{Basilico2012}}]\label{falseclaim}
Consider an instance $G$ of the \textsc{cr-uav} Problem with a single UAV.
If $G$ has a solution, then $G$ has a solution of the form $u^\omega$ 
where $u$ is a finite path through $G$ with 
$|u| \leq \displaystyle{ \frac{\max_{v \in V} \mathit{RD}(v)}{\min_{\substack{v, v' \in V \\ v \neq v'}} \mathit{FT}(v, v')} }$.
\end{claim} 

If constants are encoded in unary, the claim above would immediately imply
$\mathrm{NP}$-membership of the problem (with a single UAV).
However, the claim turned out to be incorrect, as we now give a counterexample below.
Consider the problem instance $G$ depicted in Figure~\ref{fig:cex} (we number the vertices
in clockwise order, starting with $0$ at bottom left).
The shortest possible period of a solution is $11$\footnote{This can be verified with a model checker, e.g., NuSMV~\cite{Cimatti2002}.}
whereas the claim above gives a bound of $10$.

\begin{figure}[h]
\centering
\begin{tikzpicture}[-,>=stealth', auto, node distance=2.5cm,
                    semithick, bend angle=25, every state/.style={fill=none,draw=black,text=black,shape=circle}]

\node[state, align=center]					(v1) at (-4, 1)	{\normalsize $5$};
\node[state, align=center]					(v2) at (0, 2)	{\normalsize $6$};
\node[state, align=center]					(v3) at (0, 0)	{\normalsize $9$};
\node[state, align=center]					(v4) at (-4, 3) {\normalsize $10$};

\path		(v1)	edge node [near start] {$2$} (v2)
			(v2)	edge node {$1$} (v3)
			(v3)	edge node {$2$} (v1)
			(v1)	edge node {$1$} (v4)
			(v3)	edge node [swap, near start] {$2$} (v4)
			(v2)	edge node [swap] {$2$} (v4);

\end{tikzpicture}
\caption{A periodic solution with the shortest period: $(32010230210)^\omega$}
\label{fig:cex}
\end{figure}
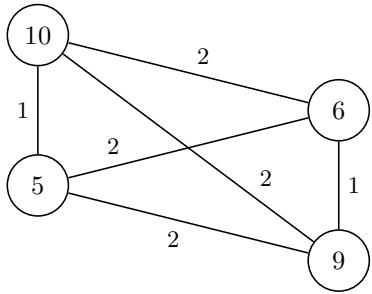

In fact, we can state a stronger result here.
The following proposition says that, the shortest period of a solution can indeed be exponential (and not linear)
in the magnitude of the largest relative deadline.

\begin{proposition}\label{prop:exponentialperiod}
There is a family of instances $\{G_n\}_{n > 0}$ (of the \textsc{cr-uav} Problem with a single UAV)
such that the shortest possible period of a solution to $G_n$
is exponential in the magnitude of the largest constant in $G_n$.\footnote{The proof of this proposition is due to Daniel Bundala.}
\end{proposition}
\begin{proof}

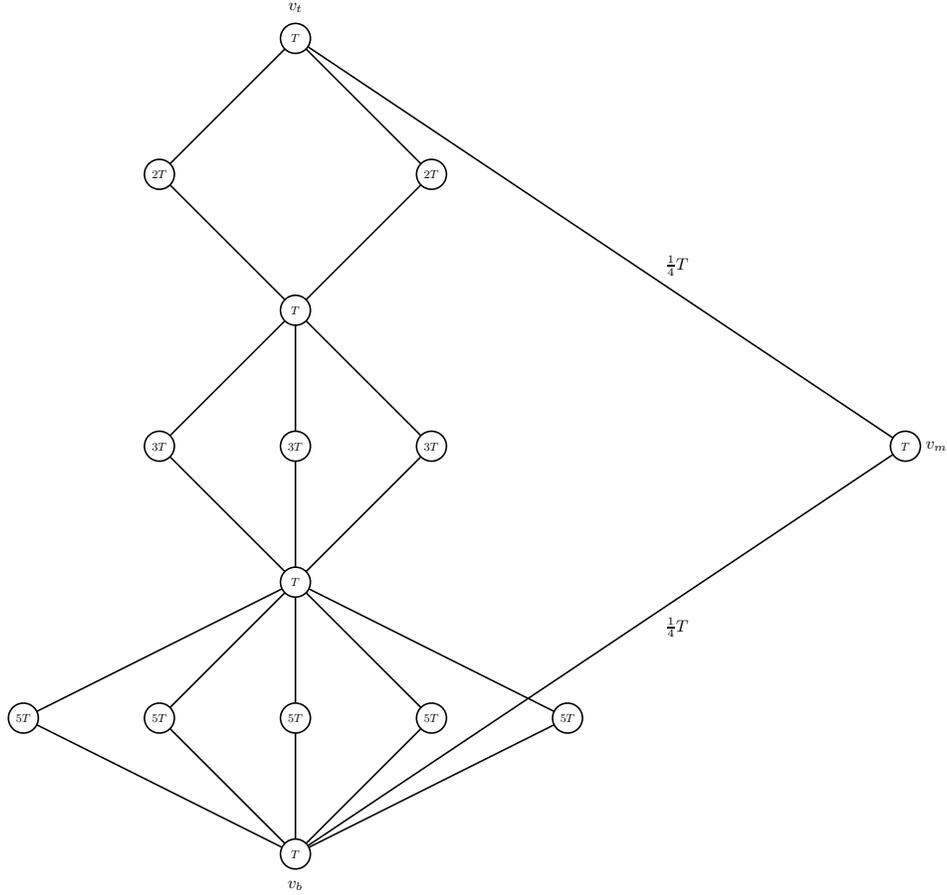
\begin{figure}[t]
\centering
\begin{tikzpicture}[-,>=stealth', auto, transform shape, node distance=.6cm, scale=0.7,
                    semithick, every state/.style={fill=none,draw=black,text=black,shape=circle, scale=0.7}]

\node[state]				(vt)	{$T$};
\node[state, draw=none]				(v0)	[below=2cm of vt] {$$};
\node[state]				(v0l)	[left=2cm of v0] {$2T$};
\node[state]				(v0r)	[right=2cm of v0] {$2T$};
\node[state]				(v1)	[below=2cm of v0] {$T$};
\node[state]				(v2)	[below=2cm of v1] {$3T$};
\node[state]				(v2l)	[left=2cm of v2] {$3T$};
\node[state]				(v2r)	[right=2cm of v2] {$3T$};
\node[state]				(v3)	[below=2cm of v2] {$T$};
\node[state]				(v4)	[below=2cm of v3] {$5T$};
\node[state]				(v4l)	[left=2cm of v4] {$5T$};
\node[state]				(v4ll)	[left=2cm of v4l] {$5T$};
\node[state]				(v4r)	[right=2cm of v4] {$5T$};
\node[state]				(v4rr)	[right=2cm of v4r] {$5T$};

\node[state]				(vb)	[below=2cm of v4] {$T$};

\node[state]				(vm)	[right=11cm of v2] {$T$};

\node [xshift=3mm] at (vm.east) {\small $v_{m}$};
\node [yshift=3mm] at (vt.north) {\small $v_{t}$};
\node [yshift=-3mm] at (vb.south) {\small $v_{b}$};

\path   (vm) edge node [swap, pos=0.4] {\small $\frac{1}{4}T$} (vt)
		  (vm) edge node [pos=0.4] {\small $\frac{1}{4}T$} (vb)
		  (vt) edge (v0l)
		  (vt) edge (v0r)
		  (v0l) edge (v1)
		  (v0r) edge (v1)
		  (v1) edge (v2l)
		  (v1) edge (v2)
		  (v1) edge (v2r)
		  (v2l) edge (v3)
		  (v2)  edge (v3)
		  (v2r) edge (v3)
		  (v3) edge (v4ll)
		  (v3) edge (v4l)
		  (v3) edge (v4)
		  (v3) edge (v4r)
		  (v3) edge (v4rr)
		  (v4ll)  edge (vb)
		  (v4l)   edge (vb)
		  (v4)    edge (vb)
		  (v4r)   edge (vb)
		  (v4rr)  edge (vb)
;

\end{tikzpicture}
\caption{The instance $G_3$}
\label{fig:prime}
\end{figure}

(\emph{Sketch.}) See Figure~\ref{fig:prime} for an illustrated example where $T = 4n$.
The $i$-th `diamond' (in top-down order) has $p_n$ branches where $p_n$ is the $n$-th prime.
The relative deadlines are set as indicated, each unlabelled edge has $FT$ set to $1$, and
each missing edge has $FT$ set to the `shortest distance' between the two relevant vertices.
It can be shown that a solution must be an infinite repetition of either (i) from $v_t$ through all the diamonds to $v_b$, to
$v_m$ and to $v_t$ again, or (ii) from $v_b$ through all the diamonds to $v_t$, to $v_m$ and to $v_b$ again.
Furthermore, in each diamond one must go straight down, and only the edges shown in the figure can be used.
It can be shown that the shortest period of a solution to $G_n$ is bounded below by $\displaystyle{\prod_{i = 1}^n p_i} = \Omega(e^n)$.
On the other hand, the number of vertices and the largest constant in $G_n$ are both $O(n^2 \ln n)$.
\end{proof}

\section{Proof of Proposition~\ref{prop:tborbt}}\label{app:tborbt}

\begin{lemma}\label{lem:toporbot}
Each segment $s_j$ must start with and end with $v_{top}$ or $v_{bot}$.
\end{lemma}

\begin{lemma}\label{lem:atleast3mp1}
The time needed from $v_{top}$ or $v_{bot}$ to any other vertex is at least $(3m + 1)l$.
\end{lemma}

\begin{lemma}\label{lem:atleast14t}
The time needed from $v_{mid}$ to any other vertex is at least $\frac{1}{4}T$.
\end{lemma}

\begin{lemma}\label{lem:segmorethanone}
Each segment $s_j$ must contain more than one vertex.
\end{lemma}
\begin{proof}
By Lemma~\ref{lem:toporbot}, without loss of generality let $s_j = v_{bot}$, a single vertex. It is easy to see
that $s_{j-1}$ must end with $v_{top}$ and $s_{j+1}$ must start with $v_{top}$, otherwise
the relative deadline of $v_{top}$ will be violated. Now consider $v_1$ (with $\mathit{RD}(v_1) = T + 2h$).
By Lemma~\ref{lem:atleast3mp1} and the fact that $\mathit{dur}(v_{top} v_{mid} v_{bot} v_{mid} v_{top}) = T$,
the relative deadline of $v_1$ is violated for sure even if $s$ visits $v_1$
immediately after $v_{top}$. This is a contradiction.
\end{proof}

\begin{proposition}\label{prop:bound}
For each segment $s_j$, $0 < \mathit{dur}(s_j) \leq \frac{1}{2}T$.
\end{proposition}
\begin{proof}
By Lemma~\ref{lem:segmorethanone} we have $\mathit{dur}(s_j) > 0$.
For the upper bound, note that $\mathit{dur}(v_{mid} s_j v_{mid}) = \frac{1}{2}T + \mathit{dur}(s_j)$
and $\mathit{RD}(v_{mid}) = T$.
\end{proof}

\begin{proposition}\label{prop:leastonce}
Each segment $s_j$ contains all vertices in $V \setminus \{ v_{mid} \}$ with relative deadlines less or equal than $T + l + 2h$.
\end{proposition}
\begin{proof}
Let $v \in V \setminus \{ v_{mid} \}$ be a vertex missing in $s_j$ with 
$\mathit{RD}(v) \leq T + l + 2h$. By Lemmas~\ref{lem:toporbot},~\ref{lem:atleast3mp1} and~\ref{lem:segmorethanone},
$\mathit{dur}(s_j) \geq 2(3m + 1)l > l + l > l + 2h$.
We have $\mathit{dur}(v_{mid} s_j v_{mid}) = \frac{1}{2}T + \mathit{dur}(s_j) > \frac{1}{2}T + l + 2h$.
By Lemma~\ref{lem:atleast14t}, $\mathit{dur}(vv_{mid} s_j v_{mid}v)$ must be greater than $T + l + 2h$
for any $v \in V \setminus \{ v_{mid} \}$, which is a contradiction.
\end{proof}

By Proposition~\ref{prop:leastonce}, we first derive a (crude) lower bound on $\mathit{dur}(s_j)$.
The sum of the minimum times needed to enter and leave every $v \in S$ and
the minimum times needed to enter and leave both ends of $s_j$ gives
\begin{equation}\label{eq:eq1}
\mathit{dur}(s_j) \geq (m-1)\big(2(3m+1)l\big) + m\big(2(3m+2)l\big) + 2(3m+1)l\,.
\end{equation}

\begin{proposition}\label{prop:tbonce}
$v_{top}$, $v_{bot}$ and each $v \in S$ appears once in each segment $s_j$.
\end{proposition}
\begin{proof}
Without loss of generality, assume one of these vertices appears more than once in $s_j$.
By a similar argument as above, we derive that $\mathit{dur}(s_j)$ is at least $(m-1)\big(2(3m+1)l\big) + m\big(2(3m+2)l\big) + 2(3m+1)l + 2(3m+1)l > \frac{1}{2}T$.
This contradicts Proposition~\ref{prop:bound}.
\end{proof}

By the proposition above, we can revise our lower bound in Eq.(\ref{eq:eq1}) by noting that
$s_j$ must start and end with different vertices. This gives
\begin{equation}\label{eq:eq2}
\mathit{dur}(s_j) \geq (m-1)\big(2(3m+1)l\big) + m\big(2(3m+2)l\big) + (3m+1)l + (3m+2)l\,.
\end{equation}

Now without loss of generality let $s_j$ ends with $v_{top}$ and $s_{j+1}$ starts with $v_{top}$.
By Eq.(\ref{eq:eq2}),
$\mathit{dur}(s_j) + \mathit{dur}(s_{j+1}) \geq 2\Big( (m - 1)\big(2(3m+1)l\big) + m\big(2(3m+2)l\big) + (3m+1)l + (3m+2)l \Big) > \frac{1}{2}T$,
and hence $\mathit{dur}(s_j v_{mid} s_{j + 1}) > T$.
By Proposition~\ref{prop:tbonce}, $v_{bot}$ can only appear at both ends of $s_j v_{mid} s_{j + 1}$,
hence its relative deadline must be violated. This is a contradiction. Proposition~\ref{prop:tborbt} is hence proved.

\section{Proof of Proposition~\ref{prop:exact}}\label{app:cnt}

%
Now we refine our lower bound in Eq.(\ref{eq:eq2}) by taking into account other vertices in variable gadgets
and consistency gadgets with $\mathit{RD}$ less or equal to $T + l + 2h$ (by Proposition~\ref{prop:leastonce}).
As many of these vertices are adjacent, we only accumulate the minimum times needed to enter them.
This gives an extra time of $m(24h + 22) + 4m + m(24h + 22)$ (note that by Proposition~\ref{prop:tbonce},
only one of the four vertices connected to a shared vertex has been entered and cannot be included
in the calculation). In total, we have
\begin{equation}\label{eq:eq3}
\mathit{dur}(s_j) \geq \frac{1}{2}T - 20m - 2h \,.
\end{equation}
 
\begin{proposition}\label{prop:32once}
Each segment $s_j$ contains all vertices with relative deadlines equal to $\frac{3}{2}T$, i.e.,
clause vertices and vertices in
\( \bigcup_{i \in \{1, \ldots, m\}} \big( (\mathit{LCG}_i \setminus \{\mathit{pvt}_i^L\}) \cup (\mathit{RCG}_i \setminus \{ \mathit{pvt}_i^R \} ) \big) \).
\end{proposition}
\begin{proof}
Assume that there is such a vertex $v$ not appearing in $s_j$.
By Eq.(\ref{eq:eq3}), we have $\mathit{dur}(v_{bot} v_{mid} s_j v_{mid} v_{top}) \geq \frac{3}{2}T - 20m - 2h$.
By Lemma~\ref{lem:atleast3mp1}, the relative deadline of $v$ must be violated as
$\mathit{dur}(v v_{bot} v_{mid} s_j v_{mid} v_{top} v) \geq \frac{3}{2}T - 20m - 2h + 2(3m + 1)l > \frac{3}{2}T$. 
This is a contradiction.
\end{proof}

Based on the previous proposition, we can further refine our lower bound on the duration of a segment.
The minimum times needed to enter
\begin{itemize}
\item clause vertices $v^{c_j}$, $j \in \{1, \ldots, h\}$
\item vertices in \( \bigcup_{i \in \{1, \ldots, m\}} \big( (\mathit{LCG}_i \setminus \{\mathit{pvt}_i^L\}) \cup (\mathit{RCG}_i \setminus \{ \mathit{pvt}_i^R \} ) \big) \)
\end{itemize}
can now be included in the calculation. We have
\begin{equation}\label{eq:eq4}
\mathit{dur}(s_j) \geq \frac{1}{2}T - 4h \,.
\end{equation}
\begin{proposition}\label{prop:pivotmoreonce}
In each segment $s_j$, each vertex in \( \bigcup_{i \in \{1, \ldots, m\}} \{\mathit{pvt}_i^L, \mathit{pvt}_i^R\} \) appears more than once.
\end{proposition}
\begin{proof}
Let there be such a vertex $v$ appearing only once in a segment.
By Lemma~\ref{lem:atleast3mp1}, there are two occurrences of $v$ in $s$ separated by at least
$\frac{1}{2} \cdot \big( \frac{1}{2}T + (\frac{1}{2}T - 4h) + \frac{1}{2}T \big) + (3m+1)l$.
This exceeds all possible values of $\mathit{RD}(v)$.
\end{proof}

By the proposition above, we assume that each vertex in \( \bigcup_{i \in \{1, \ldots, m\}} \{\mathit{pvt}_i^L, \mathit{pvt}_i^R\} \) appears twice in a segment.
Counting each such vertex once again gives an extra time of $4h$. The sum of this with Eq.(\ref{eq:eq4})
matches the upper bound in Proposition~\ref{prop:bound}.
Any more visit to a vertex in $V \setminus \{ v_{mid}, v_{top}, v_{bot}, v_1, \ldots, v_{2m-1} \}$
will immediately contradict Proposition~\ref{prop:bound}. Proposition~\ref{prop:exact} is hence proved.

\end{document}